%% file: Main.tex
\documentclass[conference]{IEEEtran}
\usepackage{pck}

\usepackage{graphicx} 

\usepackage{algorithm}
\usepackage{algorithmicx, eqparbox}
\usepackage[noend]{algpseudocode}

\usepackage[super]{nth} 
\usepackage{dirtytalk}  

\usepackage{comment} 

\usepackage{makecell}   
\usepackage{mwe}        
\usepackage{multicol}
\usepackage{multirow}
\usepackage{wrapfig, booktabs}
\usepackage{caption}    
\usepackage{subcaption}
\usepackage{amsmath}    
\usepackage{amssymb}
\usepackage[normalem]{ulem}

\usepackage{pifont}

\usepackage{soul} 
\usepackage{mathrsfs} 

\usepackage{comment}
\newif\ifsubmit

\input{macros}

\begin{document}

\title{Attack-Defense Trees with Offensive and Defensive Attributes (with Appendix)\\
}

\author{\IEEEauthorblockN{1\textsuperscript{st} Danut-Valentin Copae}
\IEEEauthorblockA{\textit{University of Twente} \\
d.v.copae@student.utwente.nl}
\and
\IEEEauthorblockN{2\textsuperscript{nd} Reza Soltani}
\IEEEauthorblockA{\textit{University of Twente} \\
r.soltani@utwente.nl}
\and
\IEEEauthorblockN{3\textsuperscript{rd} Milan Lopuhaä-Zwakenberg}
\IEEEauthorblockA{\textit{University of Twente} \\
m.a.lopuhaa@utwente.nl}
}

\maketitle

\begin{abstract}
  Effective risk management in cybersecurity requires a thorough understanding of the interplay between attacker capabilities and defense strategies. 
  Attack-Defense Trees (ADTs) are a commonly used methodology for representing this interplay; however, previous work in this domain has only focused on analyzing metrics such as cost, damage, or time from the perspective of the attacker. This approach provides an incomplete view of the system, as it neglects to model defender attributes: in real-world scenarios, defenders have finite resources for countermeasures and are similarly constrained. In this paper, we propose a novel framework that incorporates defense metrics into ADTs, and we present efficient algorithms for computing the Pareto front between defense and attack metrics.  
  Our methods encode both attacker and defender metrics as semirings, allowing our methods to be used to many metrics such as cost, damage, and skill. We analyze tree-structured ADTs using a bottom-up approach and general ADTs by translating them to binary decision diagrams. Experiments on randomly generated ADTS demonstrate that both approaches effectively handle ADTs with several hundred nodes. 
\end{abstract}


\input{Sections/1-intro}
\input{Sections/2-RelatedWork}
\input{Sections/3-ADT}

\input{Sections/4.1-PF}

\input{Sections/4.2-PF}
\input{Sections/5-Experiment}
\input{Sections/6-conclusion}


\balance
\bibliographystyle{IEEEtran}
\bibliography{bibliography}


\appendices

\section{Proof of Theorem \ref{thm:bu}}

We first prove an auxiliary lemma. Suppose that $T$ is a tree-shaped AADT whose root has two children; we call the subtrees spanned by these children $T_1$ and $T_2$. If $\mathcal{D}_i$ is the set of BDS in $T_i$, then $\mathcal{D}_1 \cap \mathcal{D}_2 = \varnothing$ due to the tree-shaped property, and $\mathcal{D}_1 \cup \mathcal{D}_2 = \mathcal{D}$. Hence $\mathbb{B}^{\mathcal{D}} \cong \mathbb{B}^{\mathcal{D}_1} \times \mathbb{B}^{\mathcal{D}_2}$. Thus we can uniquely write each element $\vec{\delta} \in \mathbb{B}^{\mathcal{D}}$ as $(\vec{\delta}_1,\vec{\delta}_2)$, with $\vec{\delta}_i \in \mathbb{B}^{\mathcal{D}_i}$. In this terminology, we have the following lemma:

\begin{lemma}
Let $T$ be a tree-shaped AADT whose root has two children; the ADTs spanned by these children are called $T_1$ and $T_2$, respectively. Let $\mathcal{D}_i$ be the set of BDS in $T_i$. Furthermore, let $\rho_1$, $\rho_2$ be the optimal attack response functions of $T_1$ and $T_2$, respectively. Then for all $\vec{\delta}_i \in \mathbb{B}^{\mathcal{D}_i}$,
\begin{equation*}
\hat{\beta}_A(\rho(\vec{\delta}_1,\vec{\delta}_2)) = \hat{\beta}_A(\rho_1(\vec{\delta}_1)) \circ_A \hat{\beta}_A(\rho_2(\vec{\delta}_2)).
\end{equation*}
\end{lemma}

\begin{proof}
We prove this statement for the case $\gamma(R_T) = \mathtt{OR}$, $\tau(R_T) = \mathtt{A}$; the other cases are similar. We first show $\hat{\beta}_A(\rho(\vec{\delta}_1,\vec{\delta}_2)) \preceq_A \hat{\beta}_A(\rho_1(\vec{\delta}_1)) \oplus_A \hat{\beta}_A(\rho_2(\vec{\delta}_2))$. To see this, note that since $R_T$ is an \texttt{OR}-gate, we have
\begin{align*}
&f_T((\rho_1(\vec{\delta}_1),\vec{0}),(\vec{\delta}_1,\vec{\delta}_2),R_T) \\
&= f_{T_1}(\rho_1(\vec{\delta}_1),\vec{\delta}_1,R_{T_1}) \vee f_{T_2}(\vec{0},\vec{\delta}_2,R_{T_2}) \\
&= 1 \vee f_{T_2}(\vec{0},\vec{\delta}_2,R_{T_2}) \\
&= 1.
\end{align*}
By definition of $\rho(\vec{\delta}_1,\vec{\delta}_2)$, this means that
\begin{align*}
\hat{\beta}_A(\rho(\vec{\delta}_1,\vec{\delta}_2)) &\preceq_A \hat{\beta}_A(\rho_1(\vec{\delta}_1),\vec{0}) \\
&= \hat{\beta}_A(\rho_1(\vec{\delta}_1)).
\end{align*}
Analogously we can show $\hat{\beta}_A(\rho(\vec{\delta}_1,\vec{\delta}_2)) \preceq_A \hat{\beta}_A(\rho_2(\vec{\delta}_2))$; since $\oplus_A$ selects the minimum w.r.t. $\preceq_A$, this implies $\hat{\beta}_A(\rho(\vec{\delta}_1,\vec{\delta}_2)) \preceq_A \hat{\beta}_A(\rho_1(\vec{\delta}_1)) \oplus_A \hat{\beta}_A(\rho_2(\vec{\delta}_2))$.

Next, we show $\hat{\beta}_A(\rho_1(\vec{\delta}_1)) \oplus_A \hat{\beta}_A(\rho_2(\vec{\delta}_2)) \preceq_A \hat{\beta}_A(\rho(\vec{\delta}_1,\vec{\delta}_2))$. Let $\rho(\vec{\delta}_1,\vec{\delta}_2) = (\vec{\alpha}_1,\vec{\alpha}_2)$. Then by definition,
\begin{align*}
1 &= f_T((\vec{\alpha}_1,\vec{\alpha}_2),(\vec{\delta}_1,\vec{\delta}_2),R_T) \\
&= f_{T_1}(\vec{\alpha}_1,\vec{\delta}_1,R_{T_1}) \vee f_{T_2}(\vec{\alpha}_2,\vec{\delta}_2,R_{T_2}).
\end{align*}
Without loss of generality, assume $f_{T_1}(\vec{\alpha}_1,\vec{\delta}_1,R_{T_1})= 1$. Then
\begin{align*}
\hat{\beta}_A(\rho_1(\vec{\delta}_1)) &\preceq_A \hat{\beta}_A(\vec{\alpha}_1) \\
&= \bigotimes_{i: \alpha_{1,i}= 1} \beta(a_i) \\
&\preceq_A \bigotimes_{i: \alpha_{1,i}= 1} \beta(a_i) \otimes\bigotimes_{i: \alpha_{2,i}= 1} \beta(a'_i) \\
&= \hat{\beta}_A(\vec{\alpha}_1,\vec{\alpha}_2) \\
&= \hat{\beta}_A(\rho(\vec{\delta}_1,\vec{\delta}_2).
\end{align*}
Here the $a_i$ are the BAS of $T_1$, while the $a'_i$ are the BAS of $T_2$. This proves $\hat{\beta}_A(\rho_1(\vec{\delta}_1)) \oplus_A \hat{\beta}_A(\rho_2(\vec{\delta}_2)) \preceq_A \hat{\beta}_A(\rho(\vec{\delta}_1,\vec{\delta}_2))$, which together with the previous inequality proves equality.
\end{proof}

We also need a second lemma to show that the semiring operations behave nicely with respect to $\sqsubseteq$.

\begin{lemma} \label{lem:pf}
Let $\circ_A$ be either $\otimes_A$ or $\oplus_A$, and let $\psi$ be the binary operation on $V_D \times V_A$ given by $\psi((x,y),(x',y')) = (x \otimes_D x',y \circ_A y')$. Then for any two subsets $X, X' \subseteq V_D \times V_A$ we have
\[
\underline{\min}_{\sqsubseteq}(\psi(X,X')) = \underline{\min}_{\sqsubseteq}(\psi(\underline{\min}_{\sqsubseteq}(X),X')).
\]
\end{lemma}

\begin{proof}
Either way we have that $y \preceq_A y'$ implies $y''  \circ_A y \preceq_A y'' \circ_A y'$. We now first prove 
\begin{equation} \label{eq:BUpf1}
\underline{\min}_{\sqsubseteq}(\psi(X,X')) \subseteq \underline{\min}_{\sqsubseteq}(\psi(\underline{\min}_{\sqsubseteq}(X),X')).
\end{equation}
Let $(x,y) \in X$ and $(x',y') \in X'$ be such that $\psi((x,y),(x',y'))$ is Pareto optimal in $\psi(X,X')$. We first prove that $\psi((x,y),(x'y')) \in \psi(\underline{\min}_{\sqsubseteq}(X),X')$. Let $(x'',y'') \in X$ be such that $(x'',y'') \sqsubseteq (x,y)$, i.e., $x'' \preceq_D x$ and $y \preceq_A y''$. Then $x'' \otimes_D x' \preceq_D x \otimes_D x'$ and $y \circ_A y' \preceq_A \circ y'' \circ_A y'$; hence
\[
\psi((x'',y''),(x',y')) \sqsubseteq \psi((x,y),(x',y').
\]
Since the RHS is Pareto optimal in $\psi(X,X')$, equality must hold; hence $\psi((x,y),(x'y')) \in \psi(\underline{\min}_{\sqsubseteq}(X),X')$. Furthermore, it cannot be dominated by elements of $\psi(\underline{\min}_{\sqsubseteq}(X),X')$ as it is not dominated by elements of $\psi(X,X')$; hence $\psi((x,y),(x',y')) \in \underline{\min}_{\sqsubseteq}(\psi(\underline{\min}_{\sqsubseteq}(X),X'))$, and equality holds in \eqref{eq:BUpf1}. Next, we prove
\begin{equation} \label{eq:BUpf2}
\underline{\min}_{\sqsubseteq}(\psi(X,X')) \supseteq \underline{\min}_{\sqsubseteq}(\psi(\underline{\min}_{\sqsubseteq}(X),X')).
\end{equation}
Let $(x,y) \in \underline{\min}_{\sqsubseteq}(X)$ and let $(x',y') \in X'$ be such that $\psi((x,y),(x',y'))$ is Pareto optimal in $\psi(\underline{\min}_{\sqsubseteq}(X),X')$. Clearly $\psi((x,y),(x',y')) \in \psi(X,X')$; we need to show that it is Pareto optimal there. Let $(z,w) \in X$, $(z',w') \in X'$ be such that $\psi((z,w),(z',w')) \in \underline{\min}_{\sqsubseteq}(\psi(X,X'))$ and such that
\begin{equation} \label{eq:BUpf3}
\psi((z,w),(z',w')) \sqsubseteq \psi((x,y),(x',y')).
\end{equation}
By \eqref{eq:BUpf1} we have that $\psi((z,w),(z',w'))$ is Pareto optimal in $\psi(\underline{\min}_{\sqsubseteq}(X),X'))$. This also holds for $\psi((x,y),(x',y'))$. Hence equality holds in \eqref{eq:BUpf3}, which in turn implies that \eqref{eq:BUpf2} holds. This completes the proof.
\end{proof}

\begin{proof}[Proof of Theorem \ref{thm:bu}]
We prove by induction that $\mathtt{BU}(T,v) = \operatorname{PF}(T_v)$ for all nodes $v$ in $T$. The statement is clear if $v$ is a leaf. Now suppose $v$ has two children $v_1$ and $v_2$, for which the induction hypothesis holds. Then
\begin{align*}
\mathtt{BU}(T,v) &= \underline{\min}_{\sqsubseteq}(\psi(\mathtt{BU}(T,v_1),\mathtt{BU}(T,v_2)))\\
&= \underline{\min}_{\sqsubseteq}(\psi(\underline{\min}_{\sqsubseteq}(\hat{\beta}(\mathbb{B}^{\mathcal{D}_1})),\underline{\min}_{\sqsubseteq}(\hat{\beta}(\mathbb{B}^{\mathcal{D}_1}))))\\
&= \underline{\min}_{\sqsubseteq}(\psi(\hat{\beta}(\mathbb{B}^{\mathcal{D}_1}),\hat{\beta}(\mathbb{B}^{\mathcal{D}_1})))\\
&= \underline{\min}_{\sqsubseteq}(\hat{\beta}(\mathbb{B}^{\mathcal{D}})\\
&= \operatorname{PF}(T_v).
\end{align*}
Here we use Lemma \ref{lem:pf} twice in the third $=$.
\end{proof}

\section{Proof of Theorem \ref{thm:bdd}}

We will prove this theorem for $\tau(R_T) = \mathtt{A}$; the case $\tau(R_T) = \mathtt{D}$ is proven analogously. Note that the definition of $\rho$, and with it those of $\hat{\beta}$ and $\operatorname{PF}(T)$, depend on $T$ only via $f_T$. Therefore, we can define $\operatorname{PF}(f)$ for any Boolean function $f$ whose variables are partitioned into $\mathcal{D} \cup \mathcal{A}$. Furthermore, for the algorithm $\mathtt{BDD_{BU}}$ we do not need the full AADT; we just need $D_{\mathcal{D}},D_{\mathcal{D}},\beta_D,\beta_A$ and $\tau(R_T)$. All other information is stored in the structure function $f_T$, which is the same as the function $f$. Consider $D_{\mathcal{D}},D_{\mathcal{D}},\beta_D,\beta_A$ and $\tau(R_T) = \mathtt{A}$ fixed, and write $\mathtt{Z}(B,w)$ for the algorithm that performs $\mathtt{BDD_{BU}}(T,B,w)$ for an appropriate AADT $T$. Then it suffices to prove the following theorem:

\begin{theorem}
Let $f$ be a Boolean function with variables $\mathcal{D} \cup \mathcal{A}$, and let $<$ be a linear order on the variables such that $d < a$ for all $d \in \mathcal{D}$ and $a \in \mathcal{A}$. Let $B$ be the ROBDD with variable ordering $<$ representing $f$. Then $\mathtt{Z}(B,R_B) = \operatorname{PF}(f)$.
\end{theorem}

\begin{proof}
We prove this by induction on the size of $B$. If $B$ has one node, then $f$ is a constant function, for which the theorem holds. If $Lab(R_B) \in \mathcal{A}$, then all variables in $B$ are in $\mathcal{A}$; in this case, the result directly follows from BDD-based analysis of semiring attack tree metrics \cite{Lopuha_Efficient_2023}. Now suppose $Lab(R_B) = d \in \mathcal{D}$. Let $B_0$ and $B_1$ be the ROBDDs spanned by the low and high child of $R_B$, respectively, representing functions $f_0$ and $f_1$ with associated functions $\rho_0,\rho_1,\hat{\beta}_0,\hat{\beta}_1$. Then $f \equiv (\neg d \wedge f_0) \vee (d \wedge f_1)$. Hence for $\vec{\delta} \in \mathbb{B}^{\mathcal{D}\setminus\{d\}}$ we get 
\begin{align*}
\hat{\beta}_D(\vec{\delta},0) &= \hat{\beta}_{D,0}(\vec{\delta}),\\
\hat{\beta}_D(\vec{\delta},1) &= \hat{\beta}_{D,1}(\vec{\delta}) \otimes_D \beta_D(d),\\
\rho(\vec{\delta},0) &= \rho_0(\vec{\delta}),\\
\rho(\vec{\delta},1) &= \rho_1(\vec{\delta}).
\end{align*}
Furthermore, write $\varphi(\vec{\delta}) := \hat{\beta}(\vec{\delta},\rho(\vec{\delta}))$; we define $\varphi_0$ and $\varphi_1$ likewise. Thus $\hat{\beta}(\mathcal{S}) = \varphi(\mathbb{B}^{\mathcal{D}})$. Then it follows that
\begin{align*} 
&\mathtt{Z}(B,R_B)\\
&= \underline{\min}_{\sqsubseteq}\Bigg(\mathtt{Z}(B_1,R_{B_1}) \\
& \ \ \cup \left\{(s \otimes_D \beta_D(d),t) \ \middle|\ (s,t) \in \mathtt{Z}(B_2,R_{B_2})\right\} \Bigg)\\
&= \underline{\min}_{\sqsubseteq}\Bigg(\underline{\min}_{\sqsubseteq}(\varphi_0(\mathbb{B}^{\mathcal{D}\setminus\{d\}})\\
& \ \ \cup \underline{\min}_{\sqsubseteq}\left\{(\hat{\beta}_{D,1}(\vec{\delta}) \otimes_D \beta_D(d),\hat{\beta}_A(\rho_1(\vec{\delta}))\ \middle|\ \vec{\delta} \in \mathbb{B}^{\mathcal{D}\setminus\{d\}} \right\} \Bigg)\\
&= \underline{\min}_{\sqsubseteq}\left(\underline{\min}_{\sqsubseteq}(\varphi(\mathbb{B}^{\mathcal{D}\setminus\{d\}} \times \{0\})) \cup \underline{\min}_{\sqsubseteq}(\varphi(\mathbb{B}^{\mathcal{D}\setminus\{d\}} \times \{1\})) \right)\\
&= \underline{\min}_{\sqsubseteq}(\varphi(\mathbb{B}^{\mathcal{D}})) = \underline{\min}_{\sqsubseteq}(\hat{\beta}(\mathcal{S})).
\end{align*}
which completes the proof by induction.
\end{proof}

\end{document}

%% file: macros.tex
\newcommand{\colorpar}[3]{\colorbox{#1}{\parbox{#2}{#3}}}
\newcommand{\marginremark}[3]{\marginpar{\colorpar{#2}{4em}{\color{#1}#3}}}

\ifsubmit
  \newcommand{\rs}[1]{}
  \newcommand{\ms}[1]{}
  \newcommand{\mlz}[1]{}
\else
  \newcommand{\rs}[1]{\marginremark{black}{yellow}{\tiny{[RS]~ #1}}}
  \newcommand{\ms}[1]{\marginremark{white}{blue}{\tiny{[MS]~#1}}}
  \newcommand{\mlz}[1]{\marginremark{purple}{white}{\tiny{[MLZ]~ #1}}}
\fi

\newcommand{\eg}{e.g., }

\algrenewcommand\algorithmicrequire{\textbf{Input:}}
\algrenewcommand\algorithmicensure{\textbf{Output:}}

%% file: Sections/1-intro.tex
\section{Introduction}

\noindent \textbf{Attack trees.}  Cyber-physical systems, such as autonomous vehicle networks or smart grids, can become notoriously complex when multiple actors are involved. 
The resulting complexity also raises the number of possible breaches that attackers can exploit, especially in systems where components rely on each other’s functioning to maintain safety. 
Consequently, there is a need for robust and systematic threat modeling systems that can cope with such attacks.
In 1999, Schneier \cite{AT} introduced attack trees (ATs), which nowadays represent one of the most prominent tools for evaluating the security of complex systems. 
Due to their simplicity and compact form, ATs are commonly used in commercial software tools as well as industrial applications, \eg analyzing the security of a SCADA system for a tank and pump facility \cite{Tanu2010AnEO} and impact analysis of electric grid feature scenarios \cite{NESCOR2015}.

The hierarchical structure of an attack tree models the root of the tree as the attacker's primary goal. 
The tree branches represent different methods the attacker could take to achieve their primary goal. 
The leaves of these branches are basic attack steps (BASs), which cannot be further refined into finer sub-goals. 
Figure \ref{fig:AT} shows an exemplary AT.
This AT includes \texttt{AND} gates, activated when all of its children are enabled, and \texttt{OR} gates, activated when only a single child is enabled.

\begin{figure}
    \centering
    \includegraphics[width=0.75\linewidth]{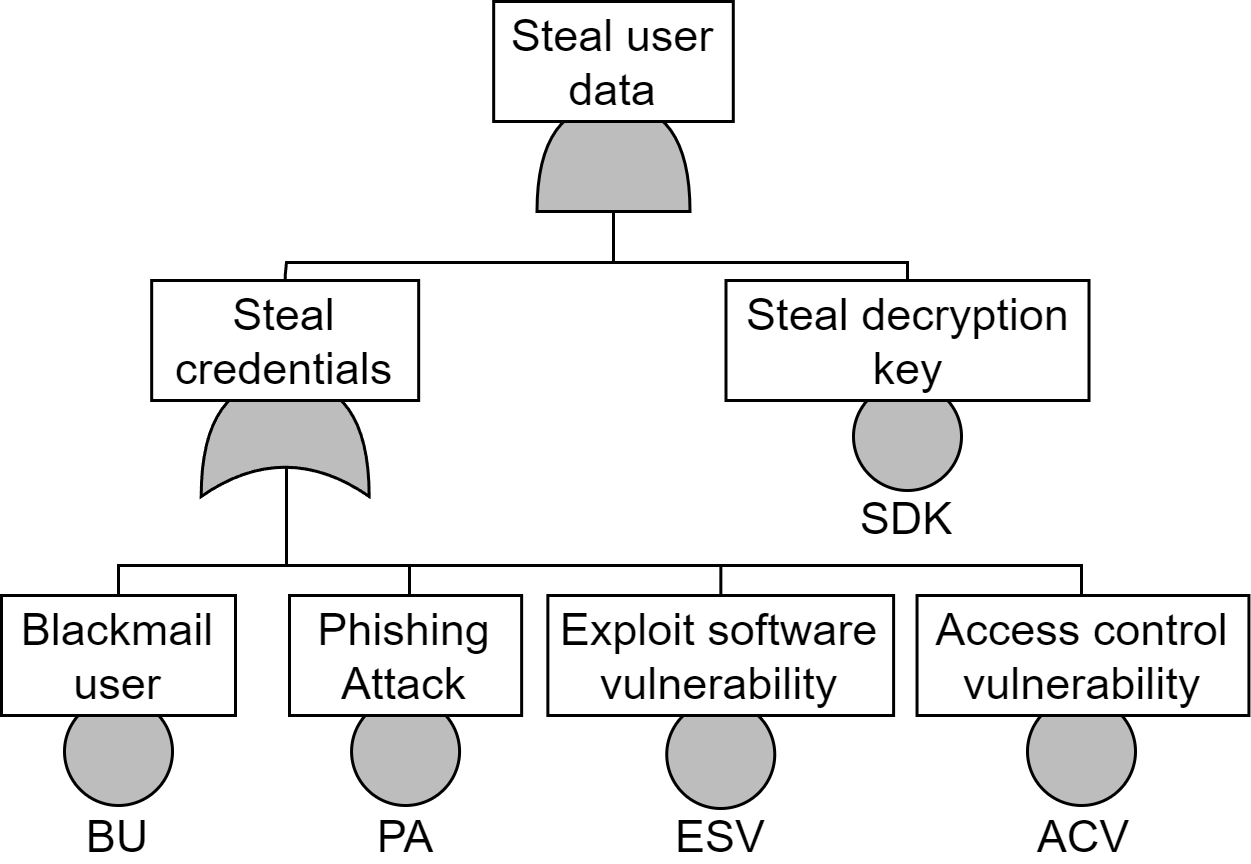}
    \caption{An AT depicting how the attacker can steal user data. To obtain the user's data, the attacker must obtain both credentials and the decryption key. The credentials can be stolen in four different ways: blackmailing the user ($BU$), conducting a phishing attack ($PA$), exploiting a software vulnerability ($ESV$), or leveraging access control vulnerabilities ($ACV$).}
    \label{fig:AT}
\end{figure}

\noindent \textbf{Atttack-defense trees.} The primary utility of ATs is to describe the various strategies an attacker can take to compromise a system through a structured decomposition of the attack into smaller objectives. 
This enables security experts to design countermeasures for preventing future attacks. 
However, one of the limitations of attack trees is that they do not account for the countermeasures implemented to prevent an attack. 
For this reason, attack-defense trees (ADTs) were introduced by Kordy et al. \cite{Barbara2011} as an extension of regular ATs to model the attacks on a system concurrently and the defenses to block those attacks. 
The defenses work by deactivating the attack nodes they are associated with, thereby disabling them. 
Figure \ref{fig:ADT} extends Fig. \ref{fig:AT} by adding counter-attack (defense) nodes. 
The basic defense steps (BDSs) are highlighted in green for better visualization. 
A defense node meets an attack node at an \texttt{INH} (inhibition) gate. 
This gate has exactly two children of opposite types, and one acts as an inhibitor. 
The \texttt{INH} gate is deactivated in the presence of an active inhibitor. On the contrary, its output equals its input when the inhibitor is not activated. 
For clarity, the edge leading to the inhibitor child of \texttt{INH} gates is marked with a small circle.

\begin{figure}
    \centering
    \includegraphics[width=0.7\linewidth]{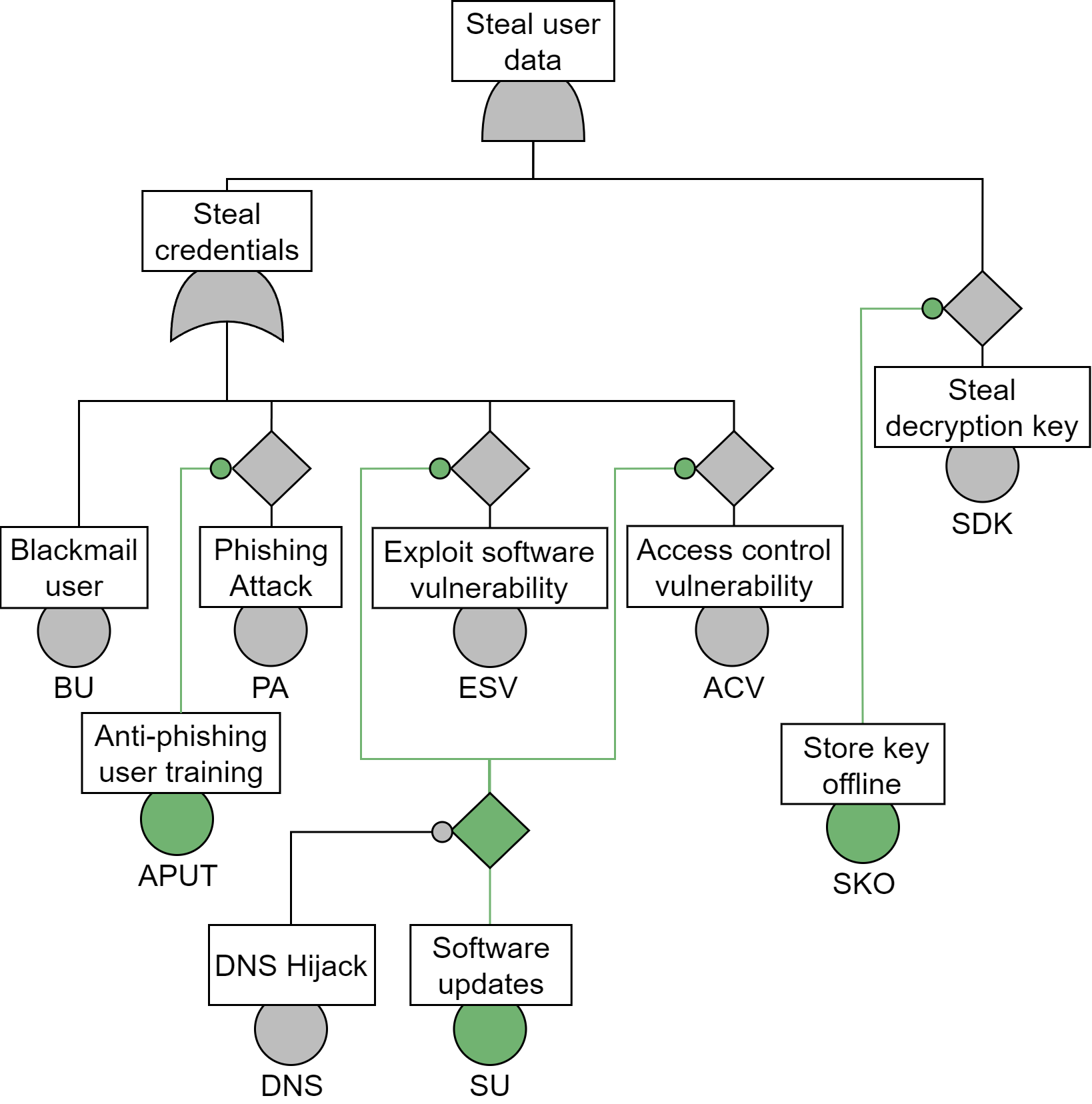}
    \caption{Attack-defense tree extending the attack tree of Fig.~\ref{fig:AT}. The defender can prevent phishing attacks ($PA$) through anti-phishing user training ($APUT$), and $SDK$ through $SKO$. Regular software updates ($SU$) prevent both $ESV$ and $ACV$. DNS Hijack ($DNS$), which does not directly contribute to reaching the top node, disables the $SU$ defense. Lastly, blackmailing the user ($BU$) has no countermeasure.}
    \label{fig:ADT}
\end{figure}

\noindent \textbf{Quantitative analysis.} 
Quantitative analysis in ADTs enables the evaluation of different strategies by assigning measurable values to actions. 
Common metrics, such as cost and probabilities, help assess the risks of attack paths and inform defensive planning.
In the current state of the attack-defense trees research, although there are two actors in an attack scenario (i.e., an attacker and a defender), only the attacker's actions are annotated with quantifiable attribute values \cite{Kordy2018,adt_pareto_domains}.
This approach fails to fully capture reality, as the defender, like the attacker, usually has finite resources.
The defender aims to select the most optimal defenses, typically those with a lower defense cost, while simultaneously making the situation more difficult for the attacker (e.g., maximizing the attacker cost). Note that the defender's and attacker's goals are conflicting: the attacker seeks to minimize their cost, whereas the defender aims to maximize attack costs with minimum defense cost. 
The defender's optimal strategy aims to reduce their own defense costs while simultaneously maximizing the difficulty for the attacker, whether in terms of cost, effort, or time. 
Also, the defender's choice are constrained by a budget.
The set of maximal achievable attacker costs or efforts for each possible defender budget forms the \emph{Pareto front}.
This leads to our research goal. 

\vspace{1.5mm}
\begin{highlighted}
\noindent \textbf{Research Goal:} 
Find efficient algorithms that compute the Pareto front between attacker and defender metrics for Attack-Defense Trees.
\end{highlighted}
\vspace{1.5mm}

By leveraging the Pareto front, defenders can make better-informed decisions to allocate resources efficiently. This enables security analysts to evaluate trade-offs between attacker and defender strategies more comprehensively, aiding in more informed and cost-effective decision-making for real-world system protection.

Most approaches, which we discuss in detail in the next section, focus on single-parameter optimization or attacker-centric metrics, leaving defender attributes and multi-objective optimization underexplored. 
Some researchers explored multi-objective and Pareto-efficient solutions for ADTs \cite{FilaWidel,Aslanyan}. Nevertheless, these works have limitations. For instance, the work of Fila and Widel \cite{FilaWidel} focuses on multi-parameter optimization within ADTs but only considers metrics that are simultaneously applicable to both attackers and defenders, such as cost or time.
Moreover, their approach fixes the metric values for one of the parties at $\infty$, thus essentially only analyzing the metrics for one party.
Similarly, the approach of Aslanyan and Nielson \cite{Aslanyan} primarily addresses Pareto fronts between attacker-centric metrics 
without extending to defender-specific attributes or the interplay between attacker and defender objectives.

In this paper, we address the limitations of existing ADT approaches by introducing a comprehensive framework that is built on formal definitions of ADTs, including their syntax and semantics, which enable the representation of attacker and defender attributes through the use of semirings. For tree-structured ADTs, we propose an efficient \texttt{Bottom-Up} (\texttt{BU}) algorithm that processes nodes iteratively from the leaves to the root, aggregating metrics to compute the Pareto front of attack and defense strategies. For more general ADTs, which include directed acyclic graph (DAG) structures, we develop a Binary Decision Diagram (BDD)-based approach to capture complex relationships between nodes and ensure scalability. Both methods were rigorously evaluated using a test suite of randomly generated ADTs with sizes up to 325 nodes, demonstrating their practicality and effectiveness. Our main contributions are:
\begin{itemize}
    \item Formal definitions of ADTs with attacker and defender attribute domains, including a semiring-based representation of metrics;
    \item A \texttt{Bottom-Up} algorithm for computing the Pareto front in tree-structured ADTs;
    \item A Binary Decision Diagram algorithm for efficiently handling the Pareto front computation for ADTs with DAG structures;
    \item Experimental validation showcasing the performance of the algorithms on large-scale ADTs.
\end{itemize}

\noindent \textbf{Paper organization:}
The paper is organized as follows. 
Section \ref{RelatedW} outlines the background of ADTs and highlights existing research gaps.
In Section \ref{ADT_formalism}, we introduce the ADT formalism, detailing its semantics, metrics, and Pareto analysis.
Section \ref{PF_tree} focuses on computing the Pareto front for tree-structured ADTs, while Section \ref{PF_DAG} extends this discussion to DAG-structured ADTs.
Finally, Section \ref{Experiments} presents experimental results and evaluates the performance of the proposed Pareto front computation methods.
Finally, Section \ref{Conclusion} concludes the paper and provides future work.

%% file: Sections/2-RelatedWork.tex
\section{Related work}\label{RelatedW}

Attack-Defense Tree (ADT) is a concept introduced by Kordy et al. \cite{Barbara2011}, which gives the system the possibility of modeling defenses through \textit{counter-attack} gates. Before this work, there had already been defined ideologies of ADTs, but they were tailored to more specific use cases. For instance, in \cite{Bistarelli2006}, defensive actions are only possible at the leaf level. Similarly, in \cite{Roy2012}, the defender can only perform counter-attacks at the leaf level but cannot have higher-level goals modeled in the tree. In their paper, Kordy et al.\cite{Barbara2011} provided a more general concept of ADTs, where attackers and defenders have equal capabilities, and counter-attacks can be modeled at intermediate nodes, including the root node. This approach offers a more comprehensive overview of the security aspects of a system.

Traditional ADT research has primarily focused on attacker-centric metrics. However, some frameworks demonstrated the value of integrating both attack and defense perspectives to better capture system dynamics \cite{afdt_paper}.
Arias et al. in \cite{hackers-vs-security} analyzed ADTs in a novel way by treating these trees as an extension of asynchronous multi-agent systems. Each node in the tree is treated as an agent that can act asynchronously. The transition functions of these nodes are then equipped with attributes. Finally, the quantitative results from the generated automata are verified with state-of-the-art model checkers such as \texttt{UPPAAL} and \texttt{Imitator}. Similarly, some methods for optimizing spare management \cite{spare_management_paper} and dynamically modeling environmental behavior for energy efficiency \cite{energy_optimization_paper,CPS} highlight the importance of considering multiple attribute domains in system modeling. Nevertheless, most analytical methods optimize one parameter at a time, such as the cost or time of an attack. However, this approach might not accurately represent complex real-world scenarios where parameters can interact (e.g., the maximum damage of an attack, given a fixed cost \cite{Milan2023_cost_damage}), potentially leading to potentially sub-optimal solutions. To analyze multiple parameters simultaneously, the leaf nodes need to be annotated with multiple values. This creates a multi-optimization problem, as there is no single solution anymore, but a set of Pareto efficient solutions called the \textit{Pareto Front}, where another does not dominate one solution in a given ordering relation \cite{adt_pareto_domains}. 
For instance, the main intuition behind this ordering, assuming the Pareto front between the attacker's damage and cost, is that if the attacker has two strategies with the same damage, but one has a lower cost, they have no incentive to choose the higher-cost strategy. 

Efforts to optimize multiple parameters in ADTs have been studied in the literature \cite{FilaWidel, Aslanyan}. The authors in \cite{FilaWidel} propose a framework for analyzing Pareto fronts in ADTs by considering multi-parameter optimization; however, their approach is limited to metrics that apply simultaneously to both attacker and defender actions, such as shared costs or time requirements. This limitation restricts its applicability to scenarios where defender-specific strategies, such as investing additional resources to impede attackers, need to be considered. Aslanyan and Nielson \cite{Aslanyan} extend ADTs to model attacker and defender interactions using a type system and propose techniques for computing Pareto-efficient solutions. Their work focuses on attacker-specific metrics, such as the cost and probability of attacks, but does not incorporate defenders' distinct metrics or strategic goals.

Our work bridges this gap by introducing a formal framework for analyzing the interplay between attacker and defender metrics in ADTs, leveraging efficient algorithms to compute the Pareto front across these metrics. Unlike prior works, our framework supports the interplay and analysis of attacker and defender strategies (metrics), offering a comprehensive view of trade-offs in security planning.

%% file: Sections/3-ADT.tex
\section{Attack-Defense Tree formalism} \label{ADT_formalism}

In this section, we introduce the formalism of the ADT, which is a structured approach to visualize and analyze the interactions between potential attacks and defensive mechanisms.
We explore ADT aspects through three key dimensions: Semantics, Metrics, and the concept of the Pareto Front.

\subsection{Semantics}

An Attack-Defense Tree (ADT) is a structured diagram that models potential security threats and the defensive actions that can counter them.
ADT semantics define the rules for interpreting an ADT, specifying how attacks and defenses interact within the structure. Our definition mostly follows \cite{Kordy2014}, though we take a graph-theoretic approach rather than a grammar approach. Like standard ATs, we have $\mathtt{AND}$- and $\mathtt{OR}$-gates. Furthermore, the effects of countermeasures on incoming attacks are modeled by Inhibition gates ($\mathtt{INH}$); these correspond to the $\mathtt{C}$-gates of \cite{Kordy2014}. Inhibition gates $v$ have two inputs: an \emph{trigger} $\bar{\vartheta}(v)$ that can stop propagation, and an \emph{inhibited} $\theta(v)$ that is required for propagation. Finally, each gate $v$ is assigned an agent $\tau(v)$, either attacker ($\mathtt{A}$) or defender ($\mathtt{D}$). Any gate type can be assigned to both agents; hence, it is also possible for defender-held inhibition gates to be inhibited by further attacker actions.

\begin{definition}[Attack-Defense Tree]\label{def:attack-defense-tree}
    An attack-defense tree is a quintuple $T = (N, E, \gamma, \tau, \vartheta)$, where $(N,E)$ is a rooted directed acyclic graph; $\gamma$ and $\tau$ are functions $\gamma\colon N \to \{\texttt{BS}, \texttt{AND}, \texttt{OR}, \texttt{INH}\},\tau: N \to \{\texttt{A}, \texttt{D}\}$; and $\vartheta$ is a function \( \vartheta: \{ v \in N \mid \gamma(v) = \texttt{INH} \} \to N \) such that \( \langle v, \vartheta(v) \rangle \in E \). Moreover, $T$ satisfies the following constraints for a node $v \in N$:
\begin{itemize}
    \item $\gamma(v) = \texttt{BS} $ if and only if $v$ is a leaf of $(N,E)$.
    \item If \( \gamma(v) = \texttt{INH} \), then $v$ has two children with different $\tau$-values, with $\tau(\bar{\vartheta}(v)) \neq \tau(\theta(v))$.
    \item if $\gamma(v) \in \{\texttt{OR}, \texttt{AND}\}$, then for all children $w$ of $v$, $\tau(w) = \tau(v)$.
\end{itemize}
\end{definition}

The root of \( T \) is denoted as \( R_T \), and the set of children of a node \( v \) as \( ch(v) = \{ w \in N \mid (v, w) \in E \} \). The set of Basic Attack Steps (\textit{BASs}) on \( T \), denoted \( \mathcal{A} \), is the set of all nodes \( v \in N \) for which \( \gamma(v) = \texttt{BS} \) and \( \tau(v) = \texttt{A} \). 
Similarly we write $\mathcal{D}\label{symbol:D}$ for the set of Basic Defense Steps (\textit{BDSs}), i.e., nodes $v$ for which $\gamma(v) = \texttt{BS}$ and $\tau(v) = \texttt{D}$. These two sets are disjoint, that is, $\mathcal{A} \cap \mathcal{D} = \varnothing$, and the union of these two sets represents the set of all basic events.
If $\gamma(v) = \mathtt{INH}$, we write $\bar{\vartheta}(v)$ for the trigger child, i.e., $ch(v) = \{\bar\vartheta(v),\theta(v)\}$. 

In an attack-defense tree, both the defender and the attacker choose a set of BDS/BAS to activate. We represent these sets as binary vectors.

\begin{definition}[Event]
An \emph{attack vector} is a binary vector $\vec{\alpha} \in \mathbb{B}^{\alpha}$. A \emph{defense vector} is a binary vector $\vec{\delta} \in \mathbb{B}^{\mathcal{D}}$. An \emph{event} is a pair $(\vec{\delta},\vec{\alpha})$ of a defense vector and an attack vector.
\end{definition}

Figure \ref{fig:attack_tree_example} is an example of a tree-structured ADT annotated with numbers representing the cost. In this illustration, the set of all attacks is $\mathcal{A} = \{a_1, a_2, a_3\}$, and the set of defenses is $\mathcal{D} = \{d_1, d_2\}$. If the attacker activates $a_2$ and $a_3$, but not $a_1$, this forms the attack vector $\vec{\alpha} = 011$. Similarly, a defensive vector where only $d_1$ is activated is represented by $\vec{\delta} = 10$. 

\begin{figure}
  \begin{center}
    \includegraphics[width=0.25\textwidth]{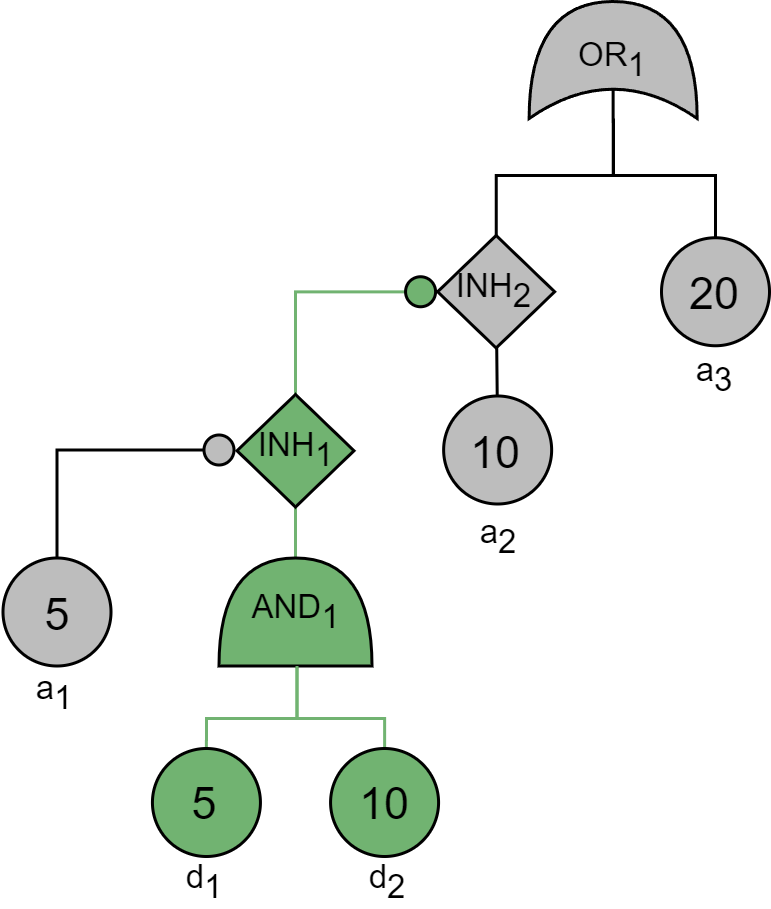}
  \end{center}
    \caption{Tree-structured ADT annotated with offensive and defensive costs.}
    \label{fig:attack_tree_example}
\end{figure}

The extent to which the attacker and defender vectors determine system status is captured by the structure function:

\begin{definition}[Structure Function] \label{def:structure-function}
Let $T$ be an ADT. The structure function of $T$ is the function $f_T \colon \mathbb{B}^{\mathcal{D}} \times \mathbb{B}^{\mathcal{A}} \times N \to \mathbb{B}$ defined as:
    
\begin{align*}
&f_{T} (\vec{\delta}, \vec{\alpha}, v)\\
&=
        \begin{cases}
          \alpha_v, & \hspace{1.0em} \text{if $v \in \mathcal{A}$}\\
          \delta_v, & \hspace{1.0em} \text{if $v \in \mathcal{D}$}\\
          \bigwedge_{w \in ch(v)} f_{T}(\vec{\alpha}, \vec{\delta}, w), & \hspace{1.0em} \text{if $\gamma(v) = \texttt{AND}$}\\
          \bigvee_{w \in ch(v)} f_{T}(\vec{\alpha}, \vec{\delta}, w), & \hspace{1.0em} \text{if $\gamma(v) = \texttt{OR}$}\\
          f_{T}(\vec{\alpha}, \vec{\delta}, \vartheta(v)) \wedge \neg  f_{T}(\vec{\alpha},  \vec{\delta}, \bar{\vartheta}(v)) &\quad \text{if  $\gamma(v) = \texttt{INH}$.}
        \end{cases}
\end{align*}
\end{definition}

\subsection{Metrics}

Security metrics, such as the lowest attack time or cost, are critical for conducting quantitative assessments of systems and making educated decisions. To achieve this, we use a well-established method called the semiring framework. 
We define security and defense attribute domains as linearly ordered unital semiring attribute domains, extending the classical semiring structure by incorporating a linear order. This additional assumption ensures compatibility with the minimization and maximization operations required for ranking and prioritizing defense and attack strategies, which are not inherently supported by traditional semirings.

\begin{definition}[Semiring Attribute Domain]
    A \textit{linearly ordered unital semiring attribute domain} (simply semiring attribute domain) is a tuple \[ L = (V, \otimes, 1_\oplus, 1_\otimes, \preceq) \]
where $V$ is a set; $1_{\otimes}$ and $1_{\otimes}$ are elements of $V$; $\otimes$ is a commutative associative binary operation on $V$; and $\preceq$ is a linear order on $V$. These furthermore satisfy the following properties:
\begin{itemize}
\item $\otimes$ is monotonous w.r.t. $\preceq$, i.e., for all $x,y,z \in V$ with $x \preceq y$ one has $x \otimes z \preceq y \otimes z$;
\item $1_{\otimes}$ is both the unit of $\otimes$, and minimal w.r.t. $\preceq$;
\item $1_{\oplus}$ is maximal w.r.t. $\preceq$.
\end{itemize}
\end{definition}

The terminology \emph{semiring} is justified as follows: given a semiring attribute domain, one can define a binary operator $\oplus$ on $V$ by
\[
x \oplus y = \min_{\preceq}(x,y).
\]
Then $(V,\oplus,\otimes)$ is a semiring; in fact, it is \emph{absorbing} in the sense of \cite{Lopuha_Efficient_2023}. We require the semiring to be ordered by $\preceq$ in order to define Pareto optimality. However, this is not a stringent constraint, as Table \ref{table:attr-domains} shows that many attack tree metrics fit into our framework.

\begin{table}[h]
    \centering
    \begin{tabular}{lcccccc}
    \toprule
    Metric & $V$ & $\oplus$ & $\otimes$ & $1_{\oplus}$ & $1_{\otimes}$ & $\preceq$ \\
    \toprule
    min cost & $[0,\infty]$ & min & + & $\infty$ & 0 & $\leq$ \\
    min time (sequential) & $[0,\infty]$ & min & + & $\infty$ & 0 & $\leq$ \\
    min time (parallel) & $[0,\infty]$ & min & max & $\infty$ & 0 & $\leq$ \\
    min skill & $[0,\infty]$ & min & max & $\infty$ & 0 & $\leq$ \\
    probability & $[0, 1]$ & max & $\cdot$ & 1 & 0 & $\geq$ \\
    \bottomrule
    \end{tabular}
    \caption{Semiring attribute domains}
    \label{table:attr-domains}
\end{table}

The attacker and defender have separate attribute domains. The attacker and defender attribute domains are described by $\mathbb{D}_A$ and $\mathbb{D}_D$, respectively. $\beta_A$ assigns an attribute value from $V_A$ to each basic attack step in $\mathcal{A}$, while $\beta_D$ does so to each basic defense step in $\mathcal{D}$. Since the attacker’s attribute values lie in \( V_A \) and the defender’s in \( V_D \). Since the attacker and defender have separate attributes, we need to combine these two into an attribute pair to perform a quantitative analysis. The attribute pair for an event will naturally have values in $V_D \times V_A$. 
\begin{definition}[Augmented Attack-Defense Tree]
   An Augmented Attack-Defense Tree (AADT) is an extension of the attack-defense tree \( T \) with associated semiring attributes. For simplicity, we use \( T \) to refer to both the original attack-defense tree and the Augmented Attack-Defense Tree (AADT), as long as the context makes the distinction clear. An AADT is defined as a tuple:

\[
T = (T, D_{\mathcal{D}}, D_{\mathcal{A}}, \beta_D, \beta_A) 
\]
Where:
\begin{itemize}
\item  \( D_{\mathcal{D}} = (V_D, \oplus_D, \otimes_D, 1_\oplus, 1_\otimes) \) is the defender's semiring attribute domain,
\item \( D_{\mathcal{A}} = (V_A, \oplus_A, \otimes_A, 1_\oplus, 1_\otimes) \) is the attacker's semiring attribute domain.
\end{itemize}
Each domain has an associated basic assignment function: \( \beta_D: \mathcal{D} \to V_D \) and \( \beta_A: \mathcal{A} \to V_A \) are functions.
\end{definition}

\begin{definition}[Metric Values]\label{def:beta-hat}
     For a given AADT, the metric value of a defense vector $\vec{\delta}$ is given by:
     \[
     \hat{\beta}_D \colon \mathbb{B}^\mathcal{D} \to V_D \text{ with } \hat{\beta}_D(\Vec{\delta}) = \underset{\delta_d=1}{\underset{d \in D}{\bigotimes}}\beta_D(d)
     \]
      while the metric value of an attack vector $\vec{\alpha}$ is given by:
      \[
      \hat{\beta}_A \colon \mathbb{B}^\mathcal{A} \to V_A \text{ with } \hat{\beta}_A(\Vec{\alpha}) = \underset{\alpha_a=1}{\underset{a \in A}{\bigotimes}} \beta_A(a)
      \]
      Lastly, the metric value of an event is given by:
      \[
      \hat{\beta} \colon \mathbb{B}^\mathcal{D} \times  \mathbb{B}^\mathcal{A} \to V_D \times V_A \text{ with } \hat{\beta}\bigl((\vec{\delta}, \vec{\alpha})\bigr) = \bigl(\hat{\beta}_D(\Vec{\delta}), \hat{\beta}_A(\Vec{\alpha})\bigr)
      \]
\end{definition}

\begin{example}
    Let's apply the definition of $\hat{\beta}$ on Fig. \ref{fig:attack_tree_example} as an example to determine the metric values of the attack and defense vector pair. 
We define a defense vector by listing the names of the activated nodes while excluding the disabled ones. Similarly, we define an attack vector in the same manner. Thus, instead of using binary vectors, we use sets to indicate the nodes that are active in each case. For example, let $11$ and $110$ be a pair of defense and attack vectors, where $11$ represents the active defense nodes (both $d_1$ and $d_2$), and $110$ represents the active attack nodes (both $a_1$ and $a_2$ are active).  Since we are working with the minimal cost domain, $\mathbb{D}_A = \mathbb{D}_D = (\mathbb{R}_{\geq 0}, \min, +)$. To determine the metric values of this vector pair, we apply the definition of $\hat{\beta}$:

\begin{align*} 
    \hat{\beta}_D\bigl(\{d_1,d_2\}\bigr) &= \underset{\alpha_d = 1}{\underset{d \in \{d_1, d_2\}}{\bigotimes}} \beta_D(d) = \beta_D(d_1) + \beta_D(d_2) \\
    &= 5 + 10 = 15 \\
    \hat{\beta}_A\bigl(\{a_1, a_2\}\bigr) &= \underset{\alpha_a = 1}{\underset{a \in \{a_1, a_2, a_3\}}{\bigotimes}} \beta_A(a) = \beta_A(a_1) + \beta_A(a_2)\\ 
    &= 5 + 10 = 15\\
    \hat{\beta}\bigl(\{d_1,d_2\}, \{a_1, a_2\}\bigr) &= (15,15)
\end{align*} 
\end{example}
Before we define the concept of a successful attack, it's important to clarify how the overall interaction between activated defenses and attacks is modeled in our framework. Each event in the AADT represents a specific attack and the corresponding defense that seeks to mitigate or prevent it. The effectiveness of a defense can vary based on the particular attacks it faces, and this interaction is what ultimately determines whether an attack is successful or not.

In our model, what is considered to be a successful attack depends critically on the root node of the tree. If the root node is an attack node, the attack is considered successful if the associated defense mechanisms fail to prevent it, which mathematically means that \(f_{T} \bigl(\vec{\delta}, \vec{\alpha}, R_T\bigr) = 1\). On the other hand, if the root node is a defense node, the attack is deemed successful if it destroys the defense mechanism, corresponding to \(f_{T} \bigl(\vec{\delta}, \vec{\alpha}, R_T\bigr) = 0\). This distinction highlights the pivotal role of the tree's structure in determining the outcome of an attack and its dependence on the defenses in place.

\subsection{Pareto analysis}

In real-world security scenarios, defenses are typically deployed before any potential attacks occur, forming a proactive line of protection. In our framework, we mirror this by first activating the defense nodes. These defenses represent the system's initial response to potential threats, and once they are in place, the attacker proceeds with their actions based on the activated defenses. This approach reflects the sequential nature of real-life interactions, where defenses are set up to mitigate risks before attacks are launched. By structuring our model this way, we can more accurately analyze the effectiveness of different defense strategies and their impact on the success or failure of subsequent attacks. 

Before defining the optimal attack response, we introduce the notation used here. The function $\texttt{arg}$ selects the attack vector that minimizes the combined metric value, allowing for any vector to be chosen if multiple vectors yield the same minimum. This minimization process reflects the attacker’s optimal response to a given defense.

\begin{definition}[Optimal Attack Response]
    The attacker's response to a defense $\vec{\delta}$ is:
    \begin{equation*}
        \rho(\vec{\delta}) = \begin{cases}
             \underset{\substack{\vec{\alpha} \in \mathbb{B}^{\mathcal{A}}\colon\\ f_T(\vec{\alpha}, \vec{\delta}, R_T) = 1}}{\text{arg }\oplus} \hat{\beta}_A(\vec{\alpha}) & \text{if $\tau(R_T)=\texttt{A}$}\\
             \underset{\substack{\vec{\alpha} \in \mathbb{B}^{\mathcal{A}}\colon\\ f_T(\vec{\alpha}, \vec{\delta}, R_T) = 0}}{\text{arg }\oplus} \hat{\beta}_A(\vec{\alpha}) & \text{if $\tau(R_T)=\texttt{D}$} \\
        \end{cases} 
    \end{equation*}
\end{definition}
In this definition, the function of the root node \( f_T(\vec{\delta}, \vec{\alpha}, R_T) \) determines the attack outcome. If the root node represents an attacker node (\(\tau(R_T) = \texttt{A}\)), the attacker succeeds if they can achieve \( f_T(\vec{\delta}, \vec{\alpha}, R_T) = 1 \). Conversely, if the root node represents a defender (\(\tau(R_T) = \texttt{D}\)), then \( f_T(\vec{\delta}, \vec{\alpha}, R_T) = 0 \) indicates that the defense has failed, resulting in an attack success.
However, there are two potential issues with this approach:
\begin{itemize}
    \item Multiple Candidates: If there are multiple attack vectors that satisfy the minimum \(\texttt{arg }\oplus\) value, any of these vectors may be chosen, as they all yield the same metric outcome for \(\hat{\beta}_A(\rho(\vec{\delta}))\).
    \item No Valid Candidates: If no valid attack vectors exist for $\vec{\delta}$, we write $\rho(\vec{\delta}) = \times$ with $\hat{\beta}_A(\times) = 1_{\oplus}$. As $1_{\oplus}$ is maximal with respect to 
\end{itemize}
This ensures that the optimal attack response is well-defined even in cases where a valid attack vector does not exist. The defender assumes that the attacker behaves optimally, i.e., always performs $\rho(\vec{\delta})$ when it exists. This leads to the following notion of \emph{feasible events}:

\begin{definition}
Let $T$ be an AADT. Its set of \emph{feasible events} $\mathcal{S} \subseteq V_D \times (V_A \cup \{\times\})$ is defined as
\[
\mathcal{S} = \{(\vec{\delta},\rho(\vec{\delta})) \mid \vec{\delta} \in \mathbb{B}^{\mathcal{D}}\}.
\]
\end{definition}

\begin{example} 
    Consider the attack-defense tree in Fig. \ref{fig:attack_tree_example}.
    For convenience, we write defense and attack vectors as binary strings; thus we write $011$ for the attack vector $\vec{\alpha} = (0,1,1)$, i.e. the attack consisting of $a_2$ and $a_3$ but not $a_1$. 
    When no defenses are active ($\vec{\delta} = 00$), the attacker can succeed by either $010$ or $001$. These have costs $10$ and $20$, respectively, and the attacker chooses the cheapest option; hence $\rho(00) = 010$. If only one defense is active, the attack responses remain the same as in the no-defense scenario, as a single defense alone is insufficient and has no effect due to the AND gate. When both defenses are active, the possible attacks are $110$ (cost 15) and $001$ (cost 20); hence $\rho(11) = 110$. Thus
    \[
    \mathcal{S} = \{(00,010),(01,010),(10,010),(11,110)\}.
    \]    
\end{example}

The trade-off between the defender's and attacker's actions can be analyzed via the Pareto Front. Specifically, the defender has two primary objectives: \say{minimizing} their own cost and \say{maximizing} the attacker’s cost, with these terms defined according to the defender's order \(\preceq_D\) and the attacker's order \(\preceq_A\), respectively. The Pareto Front represents the set of optimal trade-offs between these competing objectives. 
A point in the Pareto Front is called \textit{Pareto Optimal} if there is no other solution better in all objectives. The concept of \textit{better} is formally known as dominance, and to define it, a few pre-requisite properties must first be established.

\begin{definition}[Pareto dominant] \label{def:pareto_ordering}
    Given two events and their valuations $(s_1, t_1)$ and $(s_2, t_2)$, the pair $(s_1, t_1)$ dominates $(s_2, t_2)$ i.e. $(s_1, t_1) \sqsubseteq (s_2, t_2)$ when $s_1 \preceq_D s_2$ and $t_1 \succeq_A t_2$. 
\end{definition}

For a general poset $(X, \sqsubseteq)$, a point $x \in X$ is \emph{Pareto optimal} if it is not dominated by any other point in $X$.  Then, the Pareto frontier is the set of all Pareto optimal points in $X$, i.e. $\underbar{min}_{ \sqsubseteq}X = \left\{x \in X \mid \forall x' \in X. x' \neq x, x' \not\sqsubset x\right\}$.

\begin{example}
    Consider the ADT 
    where the defender has two possible costs, \(d \in \{5, 10\}\), and the attacker has three possible costs, \(a \in \{5, 10, 20\}\). Let \(X = \{(10, 10), (5, 20), (5, 5)\}\). The pair \((5, 20)\) dominates both \((10, 10)\) and \((5, 5)\) because it satisfies \(5 \preceq_D 10\) and \(5 \preceq_D 5\) (lower defender cost) as well as \(20 \succeq_A 10\) and \(20 \succeq_A 5\) (higher attacker cost). The Pareto front for this specific example is \((5, 20)\), representing non-dominated trade-offs between defender and attacker costs.
\end{example}

Having all the necessary mathematical prerequisites defined, we can formally state the problem statement of this paper: 

\begin{highlighted} \label{research-problem}
\textbf{Research Goal.} For an Augmented Attack-Defense Tree  $T$, we aim to find the minimal elements of the Pareto Front $\text{PF}(T)$, defined as
\[\underline{\min}_{\sqsubseteq} \hat{\beta}(\mathcal{S}) \subseteq V_D \times V_A.\]
\end{highlighted}

In principle $\operatorname{PF}(T)$ can be computed by calculating $\hat{\beta}(e)$ for all $e \in \mathcal{S}$, and computing the Pareto front of the set of pairs of metric values; thus we need to compute the  metric values of $2^{|\mathcal{D}|}$ events. In the worst case, this is unavoidable, as the following example shows.

\begin{example} \label{ex:max}
Consider the AADT $T$ of Fig.~\ref{fig:maxpf} (for a fixed integer $n$). Since $\tau(R_T) = \mathtt{D}$, the attacker's goal is to stop the defender from activating $R_T$. This is done by activating the attacks corresponding to the observed defenses, i.e., $\rho(\vec{\delta}) = \delta$ as binary vectors. Furthermore, $\hat{\beta}_D(\vec{\delta}) = \hat{\beta}_A(\vec{\delta}) = k$, where $k$ is the integer encoded by the binary number $\vec{\delta}$; hence
\[
\mathcal{S} = \{(k,k) \in \mathbb{Z}^2 \mid 0 \leq k \leq 2^n-1\}.
\]
All elements of $\mathcal{S}$ are Pareto optimal, so $|\text{PF}(T)| = 2^n = 2^{|\mathcal{D}|}$.
\end{example}

\begin{figure}
    \centering
    \includegraphics[width=0.9\linewidth]{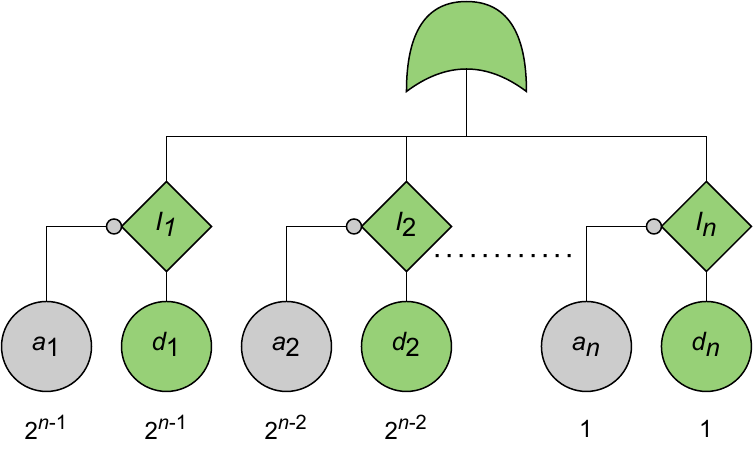}
    \caption{An AADT (with min cost as both attacker and defender metrics) with $|\text{PF}(T)| = 2^n$.}
    \label{fig:maxpf}
\vspace{-20pt}
\end{figure}

Example \ref{ex:max} shows that a worst-case exponential time complexity is unavoidable. Nevertheless, the brute force approach of computing each $\hat{\beta}(e)$ is often inefficient, as the final Pareto front is generally much smaller. In the following two sections, we present two algorithms that compute the Pareto front more efficiently, by discarding nonoptimal events along the way. Their efficiency is investigated empirically in Section \ref{Experiments}.

%% file: Sections/4.1-PF.tex
\section{The Pareto Front for Tree-shaped ADTs} \label{PF_tree}

In this section, we present methodologies for calculating the Pareto Front of tree-structured ADTs, which are ADTs where each node has a single parent, focusing on distinct structural representations. 
For tree-shaped ADTs, we employ a \texttt{Bottom-Up} Algorithm to efficiently derive optimal strategies. 
The ADT metrics are computed using a \texttt{Bottom-Up} approach based on the gate type and its function (attack or defense). \texttt{Bottom-Up} algorithms have been extensively studied for Fault Trees (FTs) and Attack Trees (ATs), providing a framework for systematic metric evaluation. To extend these algorithms to Attack-Defense Trees (ADTs), we introduce the following steps for a node $v \in N$:
\begin{enumerate}
    \item Compute the Pareto Front \texttt{Bottom-Up} for each $w \in ch(v)$.
    \item Identify all possible combinations of points from the children's Pareto Fronts. 
    \item For each combination, apply the $\min$ or $+$ operations as examples of common metrics across the defense and attack costs, depending on the values of $\gamma(v)$ and $\tau(v)$. While we use these specific metrics for illustrative purposes in this informal definition, the approach supports general metrics defined over the semiring structure.
    \item Discard the dominated points from the previous step.
\end{enumerate}

In step 1, the \texttt{Bottom-Up} algorithm is recursively applied to each child of $v$. Due to this recursive nature, the algorithm will first complete for the leaf nodes, and the results will then be propagated up the tree towards the root node. For step 2, the Cartesian product is used to find all possible ways to combine the children's Pareto fronts. Unfortunately, computing all combinations is unavoidable, as it is impossible to predict which points will be included in the Pareto front before evaluating all the value pairs in step 3.

The objective of step 3 is to transform a vector of value pairs into a single value pair $(s, t)$ for each possible combination based on Table \ref{tab:operators-applied}. 

\begin{example}\label{example5}
    Consider an AADT with the structure \( OR(INH(d_1|a_1), INH(d_2|a_2)) \) as shown in Fig. \ref{fig:example5}, where \( d_1, d_2 \) are defenses and \( a_1, a_2 \) are attacks. 
    Each defense has an associated cost: \( \beta_D(d_1) = 4 \), \( \beta_D(d_2) = 8 \). Similarly, the attack costs are \( \beta_A(a_1) = 5 \), \( \beta_A(a_2) = 10 \). In this example, we use a semiring where \((\otimes, \oplus) = (+, \min)\) for both attack and defense costs. 
\begin{enumerate}
    \item Leaf Nodes: For the leaf nodes, the cost pairs are:
    \begin{align*}
        a_1 &: (0, 5) \quad \text{(no defense cost, attack cost of 5)} \\
        a_2 &: (0, 10) \quad \text{(no defense cost, attack cost of 10)} \\
        d_1 &: (0, 0), (4, \infty) \quad \text{(defense can be inactive or active)} \\
        d_2 &: (0, 0), (8, \infty) \quad \text{(defense can be inactive or active)}.
    \end{align*}

    \item \textbf{\( INH(d_1|a_1) \):}  
    The inhibiting gate combines each \( a_1 \) cost with \( d_1 \) using \( (+, +) \), where both defense and attack costs are summed (as shown in Table \ref{tab:operators-applied}). The reason is that the attacker must successfully execute both the inhibiting attack ($d$) and the inhibited attack ($a$) for the system to fail. Similarly, the defender must allocate resources to counter both $d$ and $a$ to fully protect the system. As a result, the operations for the attack INH gate are additive, represented as $(+, +)$. Therefore:
\begin{align*}
        (0 + 0, 5 + 0) &= (0, 5) \\
        (0 + 4, 5 + \infty) &= (5, \infty).
    \end{align*}
Final Pareto front for \( INH(d_1|a_1) \): \( (0, 5), (4, \infty) \).

    \item \textbf{\( INH(d_2|a_2) \):}  
    \begin{align*}
        (0 + 0, 10 + 0) &= (0, 10) \\
        (0 + 8, 10 + \infty) &= (8, \infty).
    \end{align*}
    Final Pareto front for \( INH(d_2|a_2) \): \( (0, 10), (8, \infty) \).

    \item \textbf{\( OR(INH(d_1|a_1), INH(d_2|a_2)) \):}  
    The attack \( OR \) gate combines the Pareto fronts from \( INH(d_1|a_1) \) and \( INH(d_2|a_2) \), using \( (+, \min) \), where defense costs are summed, and attack costs take the minimum. The reason is that for the \texttt{OR} gate (\(\gamma(v) = \texttt{OR}\)) with attack function (\(\tau(v) = \texttt{A}\)), the attacker needs to succeed in only one of the available attacks to disable the system. In contrast, the defender must simultaneously defend against all possible attacks, not knowing which one the attacker might choose. Therefore:
    \begin{align*}
        (0+0, \min(5, 10)) &= (0, 5) \\
        (0+8, \min(5, \infty)) &= (8, 5) \\
        (4+0, \min(\infty, 10)) &= (4, 10) \\
        (4+8, \min(\infty, \infty)) &= (12, \infty).
    \end{align*}
    After removing dominated points, the Pareto front is: \( (0, 5), (4, 10), (12, \infty) \).
\end{enumerate}
\end{example}

\begin{figure}
    \centering
    \includegraphics[width=0.7\linewidth]{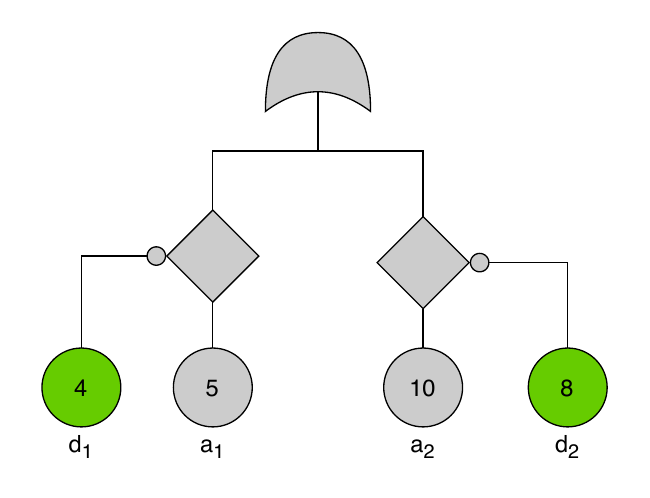}
    \caption{The AADT of Example \ref{example5}}
    \label{fig:example5}
\end{figure}

\begin{table}
    \centering
    \renewcommand{\arraystretch}{1.15} 
    \begin{tabular}{c|c|c|c}
        \toprule
        $\gamma(v)$ & $\tau(v)$ & Def. op ($\circ_D$) & Att. op ($\circ_A$) \\ 
        \hline
        \multirow{2}{*}{\texttt{AND}} & \texttt{A} & $\otimes_D$ & $\otimes_A$ \\ 
        \cline{2-4}
         & \texttt{D} & $\otimes_D$ & $\oplus_A$ \\ 
        \hline
        \multirow{2}{*}{\texttt{OR}} & \texttt{A} & $\otimes_D$ & $\oplus_A$ \\ 
        \cline{2-4}
         & \texttt{D} & $\otimes_D$ & $\otimes_A$ \\ 
        \hline
        \multirow{2}{*}{\texttt{INH}} & \texttt{A} & $\otimes_D$ & $\otimes_A$ \\ 
        \cline{2-4}
         & \texttt{D} & $\otimes_D$ & $\oplus_A$ \\ 
        \bottomrule
    \end{tabular}
    \caption{Operators applied in the \texttt{Bottom-Up} algorithm.}
    \label{tab:operators-applied}
\vspace{-15pt}
\end{table}

The output of step 3 is a set of value pairs. In step 4, we reduce the elements of this set to the Pareto Front by discarding all the dominated points according to Definition \ref{def:pareto_ordering}. In Alg. \ref{alg:BU}, the complete algorithm is presented, where $\underline{\min}_{\sqsubseteq}$ represents step 4 specifically.

    \begin{algorithm}
    \caption{\texttt{Bottom-Up}}
    \label{alg:BU}
    \begin{algorithmic}[1]
        \Require 
            \Statex{ $T$: augmented attack-defense tree}
            \Statex{$v$: node $v \in N$}
        \Ensure {...}
        \Procedure{BU}{$T, v, \beta$}
        \If{$v \in \mathcal{A}$} 
            \State \Return $\left\{(1_{\otimes}, \beta_A(v))\right\}$
        
        \ElsIf{$v \in \mathcal{D}$}
             \State \Return $\left\{(1_{\otimes},1_{\otimes}), (\beta_D(v), 1_{\oplus})\right\}$
    
        \Else 
            \State $p \gets \bigtimes_{u \in ch(v)} \text{BU}(T, u, \beta) $ 
            \State $pv \gets \{\left({\bigcirc_D}_{i=1}^n d_i, {\bigcirc_A}_{i=1}^n a_i \right) \mid (\Vec{d}, \Vec{a}) \in p \}$ 
            \State \Return \Call{$\underline{\min}_{\sqsubseteq}$}{$pv$} 
        \EndIf
        \EndProcedure
    \end{algorithmic}
    \end{algorithm}

The following theorem formalizes the relationship between the \texttt{Bottom-Up} algorithm (\(\mathrm{BU}\)) and the Pareto Front (\(\mathrm{PF}_S\)) for a tree-shaped augmented ADT (AADT).
\begin{theorem} \label{thm:bu}
Let $T$ be a tree-shaped AADT. Then $\mathrm{BU}(T,R_T) = \mathrm{PF}_{S}(T)$.
\end{theorem}
Based on this theorem, the \texttt{Bottom-Up} approach computes the Pareto Front for the root node \(R_T\), aggregating optimal points from the children nodes according to the semantics $S$. 

While the procedure loops (recursively) over the number of vertices, the main complexity comes from the fact that in Lines 7--8 we need to combine increasingly large Pareto fronts. Fig.~\ref{fig:maxpf} shows that these are worst-case exponential; hence Alg.~\ref{alg:BU} is worst-case exponential as well. Nevertheless, in practice more events will be non-Pareto-optimal, and will be discarded earlier; we will assess performance in the experiments.

\section{The Pareto Front for DAG-shaped ADTs} \label{PF_DAG}

The \texttt{Bottom-Up} (\texttt{BU}) algorithm does not work for DAG-shaped ADTs: When a node has multiple parents, the Pareto Front computed at that node is propagated multiple times up the tree, leading to that value being counted several times.

For regular ATs, we know from \cite{Lopuha_Efficient_2023} that generally, computing a metric for a semiring attribute domain in a DAG-structured AT is NP-hard. The same holds true for the ADTs in this paper since they represent an extension of regular ATs. An enumerative approach to compute the Pareto Front would be inefficient; Therefore, our approach to compute the Pareto Front for each possible defense vector $\vec{d}$, find attack vector $\vec{a}$, and then remove dominated points according to Def. \ref{def:pareto_ordering} (using $\sqsubseteq$). First, we introduce the \texttt{Naive} approach to calculate the Pareto Frontier, which provides a baseline for comparison. Then, we present our optimized solutions. Finally, we compare the results of all approaches to highlight the advantages and demonstrate how they improve upon the \texttt{Naive} method.
\subsection{\texttt{Naive} approach}
As previously mentioned, computing a metric for a semiring attribute domain in a DAG-structured ADT is NP-hard. Even so, formally defining this approach is still beneficial: it provides a practical reference point for what $\underbar{min}_{\sqsubseteq}\hat{\beta}(\mathcal{S})$ outputs for DAGs, while serving as a stepping stone for the next algorithms to improve on. 
Algorithm \ref{alg:dummy} is rather straightforward. Lines 4-11 compute $\rho(\vec{\delta})$ for each $\vec{\delta}$ by going through all the possible attacks and finding the one with the minimum metric value. In the end, the value pairs are reduced to the Pareto front using $\underline{\min}_{\sqsubseteq}$.

\begin{algorithm}
\caption{\texttt{Naive} algorithm for DAGs}
\label{alg:dummy}
\begin{algorithmic}[1]
    \Require 
        \Statex{$T$: attack-defense tree}
        \Statex{$v$: node $v \in N$}
        \Statex{$\beta$: assignment of nodes $\in N$}
    \Ensure {Pareto front of the sub-tree rooted at $v$.}
    \Procedure{\texttt{Naive}}{$T, v, \beta$}
    \State $result \gets $ new array
    \For {$\vec{\delta} \in 2^\mathcal{D}$}
        \State $att\_MetValues \gets $ new array
        \For {$\vec{\alpha} \in 2^\mathcal{A}$}
            \If{$f_T(\Vec{\delta}, \Vec{\alpha}, R_T) = 1$} 
                \State Add $\hat{\beta}_A(\vec{\alpha})$ to $att\_MetValues$
            \EndIf
        \EndFor

        \If {$att\_MetValues = \varnothing$}
            \State Add $\bigl(\hat{\beta}_D(\vec{\delta}), \infty\bigr)$ to $result$
        \Else
            \State Add $\bigl(\hat{\beta}_D(\vec{\delta}), \min_{\preceq_{A}}(att\_MetValues)\bigr)$ to $result$
        \EndIf
    \EndFor

    \State \Return $\underline{\min}_{\sqsubseteq}(result)$
    \EndProcedure
\end{algorithmic}
\end{algorithm}

%% file: Sections/4.2-PF.tex
\subsection{Binary Decision Diagrams}

\begin{figure}
\centering
\begin{minipage}{.47\columnwidth}
  \centering
  \includegraphics[width=0.7\textwidth]{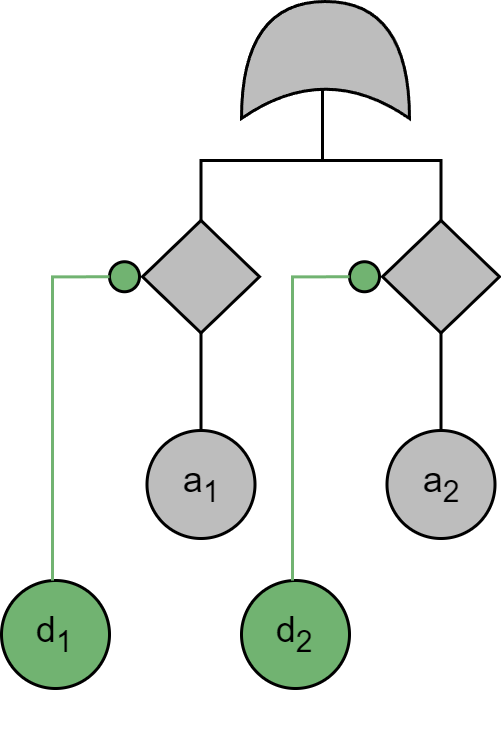}
  \label{fig:bdd_bar_order}
\end{minipage}
\hfill
\begin{minipage}{.47\columnwidth}
  \centering
  \includegraphics[width=.7\linewidth]{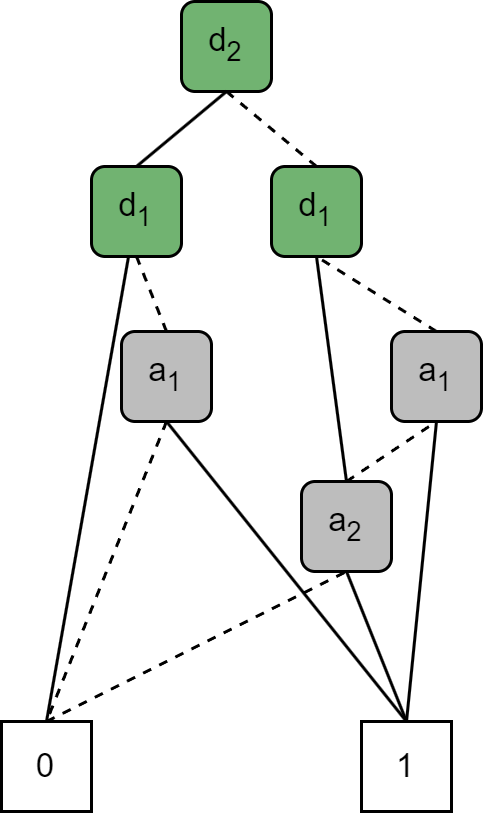}
  \label{fig:bdd_correct_order}
\end{minipage}
\caption{An ADT and a ROBDD corresponding to its structure function, with variable order $d_2 < d_1 < a_1 < a_2$. Dashed lines are labeled 0, solid lines are labeled 1.} \label{fig:aftbdd}
\vspace{-15pt}
\end{figure}

Binary Decision Diagrams (BDDs) offer a compact representation of Boolean functions. Since BDDS can have shared sub-trees, these are able to model DAG-structured ADTs. Generally, a BDD represents a Boolean function $f\colon \mathbb{B}^{vars} \rightarrow \mathbb{B}$ over a set of variables $vars$. A BDD translates a Boolean function to a flowchart, where each nonterminal node is labeled by a variable, and has two outgoing edges labeled 0 and 1. Starting from the root, given an input vector $\vec{b}$ to the Boolean function, the BDD is traversed by taking the $0$-labeled edge at a node labeled $x \in vars$ if $b_x = 0$, and taking the $1$-labeled edge if $b_x = 1$. The leaves are labeled $0$ and $1$, which denotes the output of the function for input $\vec{b}$. An example is given in Fig.~\ref{fig:aftbdd}.

We are particularly interested in \emph{reduced ordered BDDs} (ROBDDs), which means that the variable order in the BDD is the same for all paths, and that all unnecessary BDD nodes are removed. Every Boolean function has a unique ROBDD once the variable order is fixed, though different variable orderings result in different ROBDDs. Throughout the rest of this section, we will simply write `BDD' for `ROBDD'.

\begin{definition} 
    An \textit{ROBDD} is a tuple $B = (W, Low, High, Lab, \prec)$ over a set of variables $vars$, where:
    \begin{itemize}
        \item The set of nodes $W$ is partitioned into terminal nodes $(W_t)$ and nonterminal nodes $(W_n)$;
        \item $Low : W_n \rightarrow W$ maps each nonterminal node to its \textit{low} child;
        \item $High : W_n \rightarrow W$ maps each nonterminal node to its \textit{high} child;
        \item $Lab : W \rightarrow \{0, 1\} \cup vars$ maps terminal nodes to Booleans and nonterminal nodes to variables:
        \[
        Lab(w) \in 
        \begin{cases} 
          \{0_T, 1_T\} & \text{if } w \in W_t, \\
          vars & \text{if } w \in W_n;
        \end{cases}
        \]
        \item $\prec$ is a strict total order on $vars$ (called the variable ordering).
    \end{itemize}
    Moreover, $B$ satisfies the following constraints:
    \begin{itemize}
        \item $(W, E)$ is a connected Directed Acyclic Graph (DAG), where
        \[
        E = \{(w, w') \in W^2 \mid w' \in \{Low(w), High(w)\}\};
        \]
        \item $B$ has a unique root, denoted $R_B$:
        \[
        \exists! R_B \in W. \forall w \in W_n . R_B \notin \{Low(w), High(w)\};
        \]
        \item The variable ordering $\prec$ is respected, i.e., for every edge $(w, w') \in E$, if $w \in W_n$ and $w' \in W_n$, then $Lab(w) \prec Lab(w')$;
        \item $B$ is \textit{reduced}, meaning:
        \begin{enumerate}
            \item No two distinct nodes have the same label, low child, and high child:
            \begin{align*}
            \forall w, w' \in W_n,\; w \neq w' \implies & \big(Lab(w) \neq Lab(w') \\
            & \lor Low(w) \neq Low(w') \\
            & \lor High(w) \neq High(w')\big).
            \end{align*}

            \item No node has identical low and high children:
            \[
            \forall w \in W_n . Low(w) \neq High(w).
            \]
        \end{enumerate}
    \end{itemize}
\end{definition}

Conventionally edges $(w,Low(w))$ are labeled $0$, and $(w,High(w))$ are labeled $1$; as described above, this determines the Boolean function $B$ represents.

BDDs have been used for the quantitative analysis of fault trees \cite{rauzy-fault-trees} and attack trees \cite{Lopuha_Efficient_2023}. This makes them a promising candidate for ADTs as well. One important consideration is the choice of the variable order $<$. For fault trees, attack trees, and ADTs (as we will see next section), the size of the BDD is a major factor in algorithm complexity. BDD size is worst-case exponential; however, in practice it heavily depends on the choice of $<$. At the same time, finding the optimal $<$ is NP-hard. These two facts have lead to the rich field of BDD minimization, where heuristic methods are developed to find good variable orders \cite{ebendt2005advanced}.

For ADTs the variable order is more restricted: we demand that in the BDD, defenses come before attacks. This reflects that the attacker can choose their attack upon having observed the defender's actions. This turns out to be necessary in order for our algorithm to work (see Theorem \ref{thm:bdd}).

\begin{definition}\label{def:defense-first-ordering}
    A defense-first ordering is a linear order $<$ on $\mathcal{D} \cup \mathcal{A}$ such that for all $d \in \mathcal{D}$ and $a \in \mathcal{A}$, $d < a$.
\end{definition}

\begin{example} \label{exa:bdd}
Consider the AFT of Fig.~\ref{fig:aftbdd}. The variable order of the BDD is defense-first. The paths in the BDD correspond to evaluations of the structure function; for instance, the path $d_2 \rightarrow d_1 \rightarrow a_1 \rightarrow a_2 \rightarrow 1$ represents the fact that $f_T(10,01) = 1$. The path $d_2\rightarrow d_1 \rightarrow a_1 \rightarrow 0$ represents the fact that $f_T(10,0*) = 0$; here $*$ denotes the fact that the value of $\alpha_{a_2}$ is irrelevant.
\end{example}

\subsection{BDD-based algorithm for ADT}

We can use an ROBDD corresponding to an AAFT $T$ to compute its Pareto front. Suppose $\tau(R_T) = \mathtt{A}$, so that the attacker's goal is to reach the 1-leaf of the ROBDD. Every path $\pi$ from the root to 1 represents an event $e_{\pi} = (\vec{\delta},\vec{\alpha})$, with $\delta_d = 1$ if and only if there is a node $w$ such that $Lab(w) = d$ and $(w,High(w))$ is part of $\pi$; the vector $\vec{\alpha}$ is defined likewise. Thus, in the notation of Example \ref{exa:bdd}, we set $* = 0$. Not all possible events are present as paths; for instance in Fig.~\ref{fig:aftbdd} the event $(01,11)$ is not present. However, all feasible events are in the ROBDD; hence
\begin{equation*}
\operatorname{PF}(T) = \underline{\min}_{\sqsubseteq} \{\hat{\beta}(e_{\pi}) \mid \pi \text{ path from root to 1 in ROBDD}\}.
\end{equation*}
Typically, the ROBDD will represent some nonfeasible events as well. For instance, in Fig.~\ref{fig:aftbdd}, the events $(00,10)$ and $(00,01)$ are both present, but only one of them will be feasible. However, this is not a problem, as nonfeasible events will be filtered out in the $\underline{\min}_{\sqsubseteq}$ of the equation above.

This suggests an algorithm to compute $\operatorname{PF}(T)$: compute all paths from the root to 1, determine their metric values, and take the Pareto front. However, this is inefficient as there will be many, partially overlapping parts. Instead, we perform a bottom-up computation on the BDD. This is similar to a typical shortest path algorithm in a directed acyclic graph, except that (1) our computations are in the semirings $D_D,D_A$ instead of $\mathbb{R}$ and (2) instead of propagating a single value, we propagate a Pareto front of optimal value pairs. If $T$ has no defense nodes (i.e., an attack tree), our algorithm is identical to the BDD-based attack tree analysis algorithm of \cite{Lopuha_Efficient_2023}.

\begin{algorithm}
\caption{BDD Bottom Up ($\texttt{BDD}_\texttt{BU}$)}
\label{alg:bdd-bu}
\begin{algorithmic}[1]
    \Require 
        \Statex{$T = (T,D_{\mathcal{D}},D_{\mathcal{A}},\beta_D,\beta_A)$: augmented ADT}
        \Statex{$B$: ROBDD representing $f_T$}
        \Statex{$w$: node of $B$}
    \Ensure {Pareto Front of the ADT encoded by $B$.}
    \Procedure{$\texttt{BDD}_\texttt{BU}$}{$T,B,w$}

    \If{$Lab(w)=0$}
        \State \Return $\begin{cases}
        \{(1_{\otimes_D},1_{\oplus_A})\} & \text{ if $\tau(R_T) = \mathtt{A}$};\\
        \{(1_{\otimes_D},1_{\otimes_A})\} & \text{ if $\tau(R_T) = \mathtt{D}$};
        \end{cases}$
    \ElsIf{$Lab(w)=1$}
        \State \Return $\begin{cases}
        \{(1_{\otimes_D},1_{\otimes_A})\} & \text{ if $\tau(R_T) = \mathtt{A}$};\\
        \{(1_{\otimes_D},1_{\oplus_A})\} & \text{ if $\tau(R_T) = \mathtt{D}$};
        \end{cases}$
    \ElsIf{$Lab(w) \in \mathcal{A}$}
        \State $\{(1_{\otimes_D},u_{0})\} \leftarrow \mathtt{BDD_{BU}}(T,B,Low(w))$;
        \State $\{(1_{\otimes_D},u_{1})\} \leftarrow \mathtt{BDD_{BU}}(T,B,High(w))$;
        \State \Return $\{(1_{\otimes_D},u_{low}\oplus_A(\beta_A(Lab(w))\otimes_A u_{high}))\}$;
    \Else
        \State $P_{0} \leftarrow \mathtt{BDD_{BU}}(T,B,Low(w))$;
        \State $P_{1} \leftarrow \mathtt{BDD_{BU}}(T,B,High(w))$;
        \State $P \leftarrow P_{0} \cup \{(\beta_D(Lab(w))\otimes_D u,u') \mid (u,u') \in P_{1}\}$;
        \State \Return $\underline{\min}_{\sqsubseteq}(P)$
    \EndIf

    \EndProcedure
\end{algorithmic}
\end{algorithm}

The algorithm $\texttt{BDD}_\texttt{BU}$ is presented in Alg.~\ref{alg:bdd-bu}. It is a recursive algorithm; we are interested in $\texttt{BDD}_\texttt{BU}(T,B,R_B)$. To explain it, we will use the min cost semiring (see Table \ref{table:attr-domains}) for both the attacker and defender, but the algorithm works for all semiring attribute domains; we also assume that $\tau(R_T) = \mathtt{A}$; i.e., the attacker's goal is to reach the leaf 1. At every node $w$, each element  $(u,u') \in \texttt{BDD}_\texttt{BU}(T,B,w)$ represents the fact that the attacker can reach 1 from node $w$ by spending $u'$, provided the defender has budget at most $u$. How $\texttt{BDD}_\texttt{BU}(T,B,w)$ is computed depends on the label $Lab(w)$:
\begin{itemize}
    \item If $Lab(w) = 0$, then $w$ is the 0-leaf. From here it is impossible to reach the 1-leaf, i.e., even if the defender has no budget. This is represented by the value pair $(0,\infty)$ (Line 3). 
    \item If $Lab(w) = 0$, then $w$ is the 1-leaf. Thus the attacker does not need to spend any cost to get to 1, which is represented by the value pair $(0,0)$ (Line 5).
    \item If $Lab(w) \in \mathcal{A}$, we compute $\texttt{BDD}_\texttt{BU}(T,B,w)$ from its children $Low(w)$ and $High(w)$. We assume that both Pareto fronts $\texttt{BDD}_\texttt{BU}(T,B,Low(w))$ and $\texttt{BDD}_\texttt{BU}(T,B,High(w))$ consist of a single element with first coefficient $1_{\otimes_D} = 0$ (Lines 7-8); this is true because it is true for the leaves and for the output of $\mathcal{A}$-labeled nodes (Line 9), and our variable order ensures that no descendants of $w$ are $\mathcal{D}$-labeled. At node $w$, the attacker can choose to either not perform attack $Lab(w)$, move to $Low(w)$, and perform the optimal attack there, which has cost $u_{0}$; or perform attack $Lab(w)$ (paying cost $\beta_A(Lab(w))$), move to $High(w)$, and perform the optimal attack there, with cost $u_{1}$. The attacker chooses between the two by minimizing cost, which is $\min(u_{low},\beta_A(Lab(w))+u_{high})$ (Line 9). Since the defender plays no role here, the defender's metric value stays $1_{\otimes_D} = 0$.
    \item If $Lab(w) \in \mathcal{D}$, the defender can choose between either not activating $Lab(w)$, moving to $Low(w)$, and picking one of the defense options there; or activating $Lab(w)$ (incurring cost $\beta_D(Lab(w))$, moving to $High(w)$, and picking one of the defense options there. The resulting combination of defense-attack metric value pairs is given in Line 13. Not all options will be optimal, so we only keep its Pareto front (Line 14).
\end{itemize}

The following theorem expresses correctness.

\begin{theorem} \label{thm:bdd}
Let $T$ be an AADT, and let $<$ be a defense-first ordering on $\mathcal{A} \cup \mathcal{D}$ Let $B$ be the ROBDD with variable ordering $<$ representing $f_T$. Then $\texttt{BDD}_\texttt{BU}(T,B,R_B) = \text{PF}(T)$.
\end{theorem}

The complexity is $\mathcal{O}(|W|p^2)$, where $W$ is the set of nodes of the BDD, and $p$ is the maximal size of the Pareto fronts involved in the computation. Both are worst-case exponential in the size of the ADT (see Fig.~\ref{fig:maxpf}); however, we expect both to be reasonably small in practice, resulting in fast computation. We will test this expectation in the experiments.

%% file: Sections/5-Experiment.tex
\section{Experiments} \label{Experiments}

We empirically evaluate our methods, using a case study adapted from \cite{Kordy2018}, as well as 120 randomly generated ADTs.

\subsection{Case Study: Money Theft} \label{ssec:case}

We apply our methods to a real-world ADT presented by Kordy and Wideł \cite{Kordy2018}  modelling a money theft scenario, either via stealing someone's card and using an ATM, or via online banking (see Figure~\ref{fig:kordy}). That work assigns a unitless value to each BAS, representing the attacker's cost of performing that action. We also assign cost values to the BDS.

As mentioned in Section \ref{RelatedW}, existing work does not compute the Pareto front between attack and defense costs. Instead, the minimal attack cost discussed in \cite{Kordy2018} only considers attacks that cannot be prevented. This amounts to giving the defender infinite budget, and to just one point on the Pareto front.

This is a DAG-shaped ADT, since \emph{Phishing} has two parents. As such, we cannot apply our bottom-up approach directly. Instead, we assume that \emph{Phishing} needs to be performed twice in order to activate both \emph{Get Password} and \emph{Get username}. This turns the ADT into a tree-shaped one, and we can perform the \texttt{Bottom-Up} algorithm. This is inscribed in red in Fig.~\ref{fig:kordy}. As we can see, the final Pareto Front contains of only 3 pairs: $(0,90)$ is the cheapest attack on \emph{Via ATM}. The defender can thwart this via \emph{Cover Keypad}; if they do this, the attacker instead takes the cheapest attack on \emph{Via Online Banking}, represented by the pair $(30,150)$. If the defender furthermore activates \emph{SMS Authentification}, it is most advantageous for the attacker to attack \emph{Via ATM} again, disabling the defender's action using \emph{Camera}. This is represented by the Pareto-optimal pair $(50,165)$. Note that the BDS \emph{Strong Pwd} is not part of any Pareto-optimal point, suggesting that this action does not help the defender and should be avoided. In \cite{Kordy2018} the outcome of the analysis is 165, which corresponds to our final Pareto-optimal pair; thus our analysis gives a more complete picture of the interplay between attacker and defender, by showing the defender the effect of varying their budget on overall security.

We also compute the Pareto front using $\texttt{BDD}_{\texttt{BU}}$, which accurately analyzes the DAG-shaped ADT. We get the Pareto-optimal pairs $(0,80)$, $(20,90)$, $(50,140)$. Again, we see that the Pareto front is much smaller than the exponential upper bound would suggest. The optimal strategies are different to the tree-shaped case: here $\{$\emph{Phishing}, \emph{Log In \& Execute Transfer}$\}$ is optimal if the defender has no budget. 140 is also the metric computed in \cite{Kordy2018} under so-called \emph{set semantics}, which model DAG-shaped ADTs. Again existing work only gives a single value, instead of the whole Pareto front.
Although this case study demonstrates the applicability of our method, it abstracts away factors like dynamic system behavior, uncertainty in metrics, and organizational constraints, which would need to be addressed when scaling to more complex, real-world scenarios.

\begin{figure*}
  \begin{center}
    \includegraphics[width=0.90\textwidth]{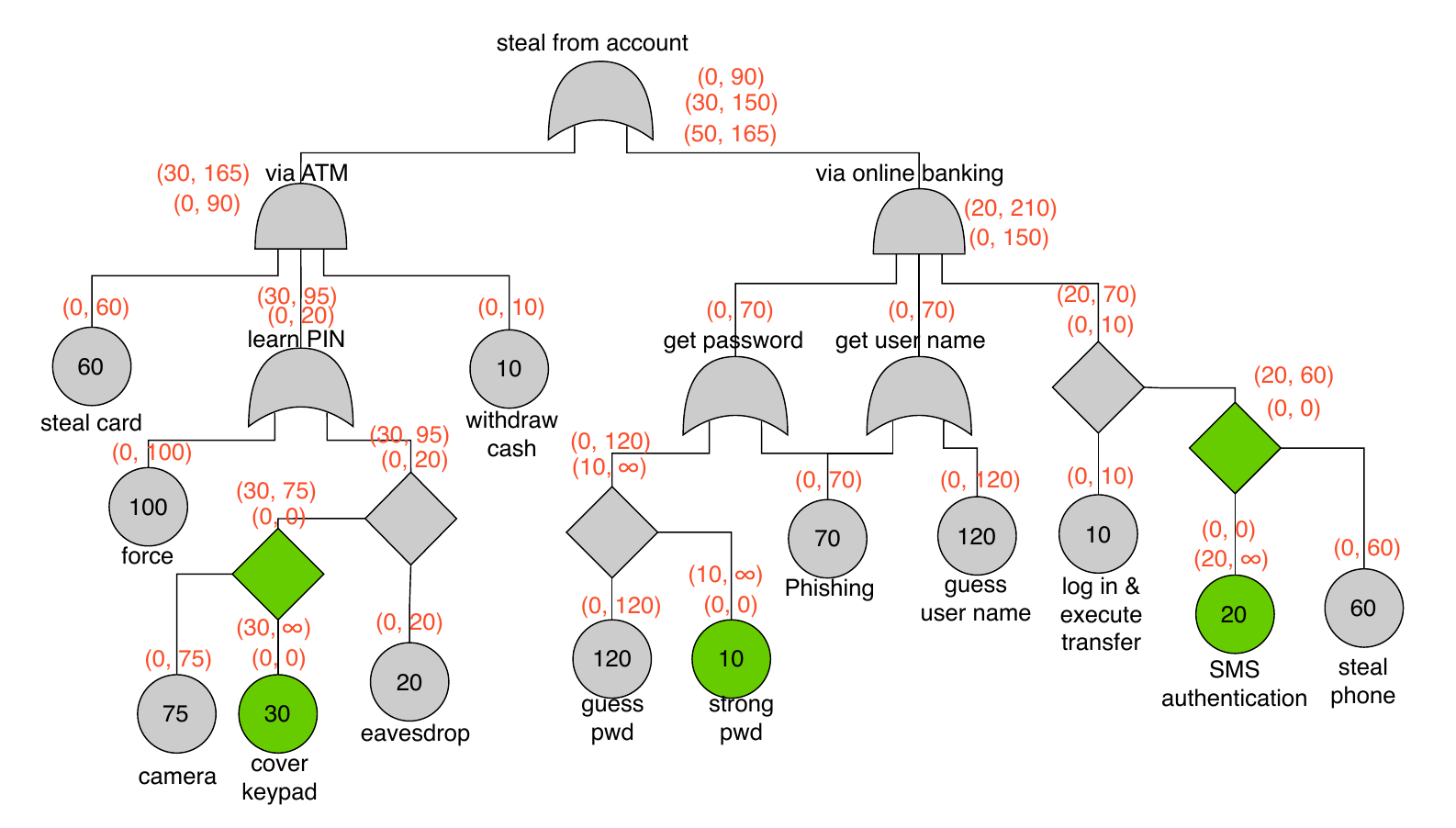}
  \end{center}
  \vspace{-1em}
    \caption{Attack-defense tree representing a money theft scenario, adapted from \cite{Kordy2018}. Values inscribed in BAS/BDS are attacker/defender costs; red values are the Pareto fronts at every node, computed \texttt{Bottom-Up}.}
    \label{fig:kordy}
\end{figure*}

\begin{figure}
  \begin{center}
    \includegraphics[width=0.40\textwidth]{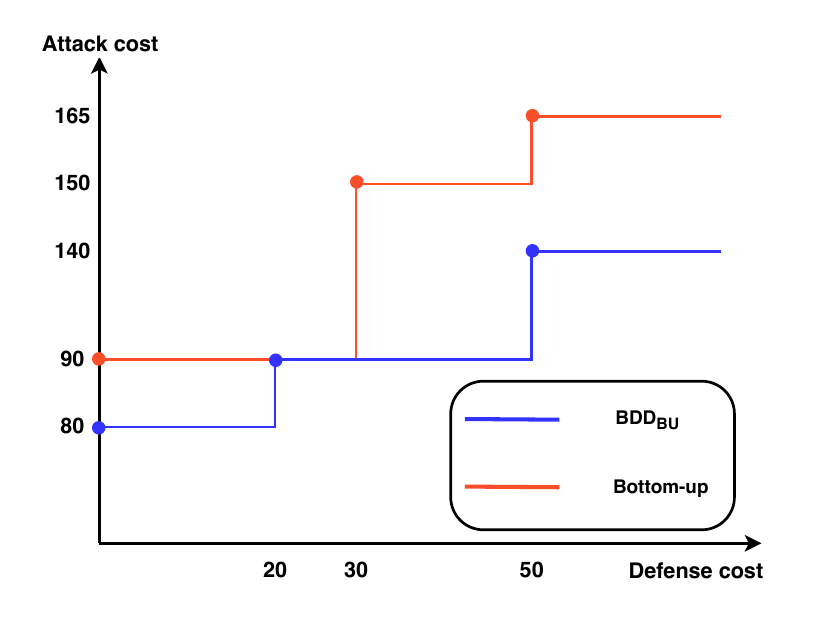}
  \end{center}
  \vspace{-1.5em}
    \caption{Pareto front for the ADT of Fig.~\ref{fig:kordy}, both under \texttt{Bottom-Up} and $\texttt{BDD}_{\texttt{BU}}$  analysis.}
    \label{fig:kordyPF}
    \vspace{-1.5em}
\end{figure}

\subsection{Synthetic ADTs} \label{ssec:synth}

The performance of \texttt{BU} and $\texttt{BDD}_{\texttt{BU}}$ is evaluated against the \texttt{Naive} approach. The \texttt{Naive} algorithm computes $\rho(\vec{\delta})$ for each defense vector $\vec{\delta}$ by iterating over all possible attacks and selecting the one with the minimum metric value, followed by reducing the resulting value pairs to the Pareto front using $\underline{\min}_{\sqsubseteq}$. 
Since existing work does not account for the interplay between attacker and defender metrics, direct comparisons with prior literature are not feasible. 

One approach to overcome this is to annotate the existing ADTs found in the literature with cost values for the defender. We were able to find several ADTs with $25 \leq |N| \leq 50$ \cite{adt_homeland_security, DontMissTheForest2014} but only a limited number of ADTs with $50 \leq |N| \leq 100$ \cite{Fraile2016, Bagnato2012, adt_online_banking}. Finding ADTs with $|N| \geq 100$ can be challenging as they are typically not made public for confidentiality reasons \cite{Paul_2014}. 

In practice, the size of attack trees can range from a few dozen to several hundred nodes \cite{vigo-automated-generation}. Given the relatively small number and size of ADTs found in the literature, we do not consider this to be a large enough testing suite to evaluate algorithms. Consequently, it is necessary to synthetically generate ADTs. Two common techniques for generating ATs are combining literature trees into a single one \cite{Lopuha2023_AttackTimeAnalysis}, or generating random ATs from scratch. We focused on the latter to cover a wider range of scenarios and create a more robust test suite. 

A risk analysis algorithm that takes several days may be feasible for some applications. However, within the context of this work, given the hardware limitations and restricted time to conduct experiments, we limit our testing scope by not pursuing computations that take more than $10^4$ seconds.

All experiments were performed on a machine with an Intel Core i5-12600K 3.7Ghz processor and 16GB of RAM. The algorithms are implemented in Python 3.12. Although faster BDD run times could perhaps be achieved using a C implementation such as in Sylvan \cite{dijk2016sylvan} or CUDD \cite{somenzi2015cudd}, we opted to maintain a consistent testing environment for all algorithms. The code and results are available on GitHub. \footnote{\url{https://github.com/dvcopae/thesis_adtrees}}.

We employ a large number of randomly generated ADTs for a statistically significant comparison of the algorithms' performance. After setting a maximum number of children $n$, nodes with random proprieties (gate type, attack/defense type, number of children) are recursively generated until the tree contains $n$ nodes (see the Appendix
). This approach naturally creates tree- and DAG-structured ADTs. 

Fig. \ref{fig:scatter} presents pairwise comparisons between the algorithms. A summary of these comparisons is given in \autoref{fig:pairwise-summary}. Algorithm runtimes across all random graphs are aggregated by taking the median at each interval of $|N| = 20$. Since the run time of \texttt{Naive} increases at an exponential rate, the values at the end of the interval will be drastically different than those at the beginning. To better represent the central tendency of the interval, the median is used instead of the average. 

Regarding the \texttt{Naive} algorithm, it is notable that for certain small-sized trees with fewer than 50 nodes, it surprisingly outperforms $\texttt{BDD}_{\texttt{BU}}$. We hypothesize that this is because, at a very small number of nodes, the time required to construct the BDD model constitutes a significant proportion of the total run time. However, as the number of nodes increases, the \texttt{Naive} algorithm approaches a runtime of $10^4$ seconds, even for some trees with less than 50 nodes. This algorithm has the slowest performance among the considered methods.

\begin{figure}[t]
    \centering
    \begin{minipage}{0.7\columnwidth}
        \centering
        \includegraphics[width=\textwidth]{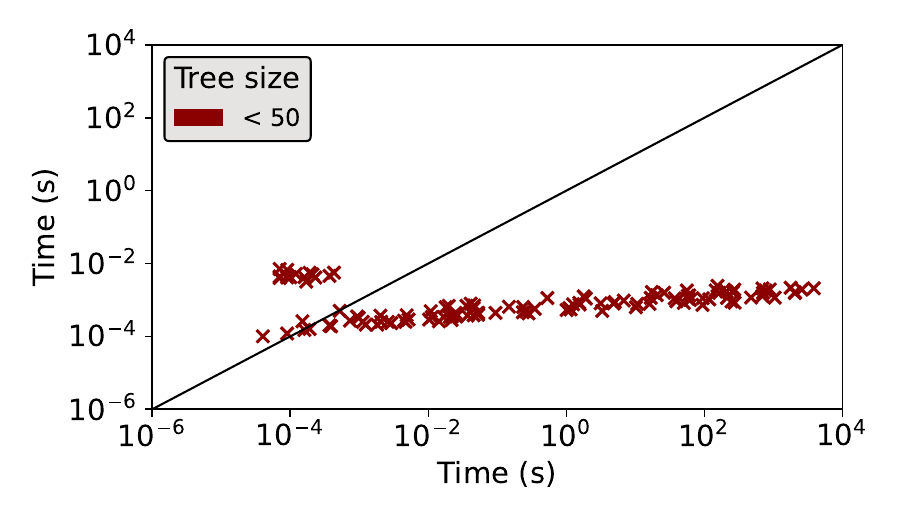}
        \caption*{\footnotesize (a) \texttt{Naive}, $\texttt{BDD}_{\texttt{BU}}$}
    \end{minipage}
    \hfill
    \begin{minipage}{0.7\columnwidth}
        \centering
        \includegraphics[width=\textwidth]{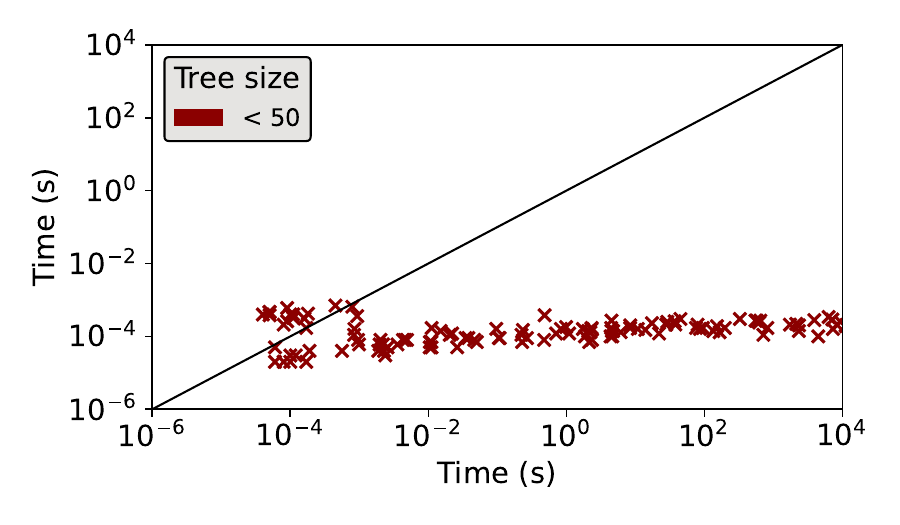}
        \caption*{\footnotesize (c) \texttt{Naive}, \texttt{BU}}
    \end{minipage}
    \hfill
    \begin{minipage}{0.7\columnwidth}
        \centering
        \includegraphics[width=\textwidth]{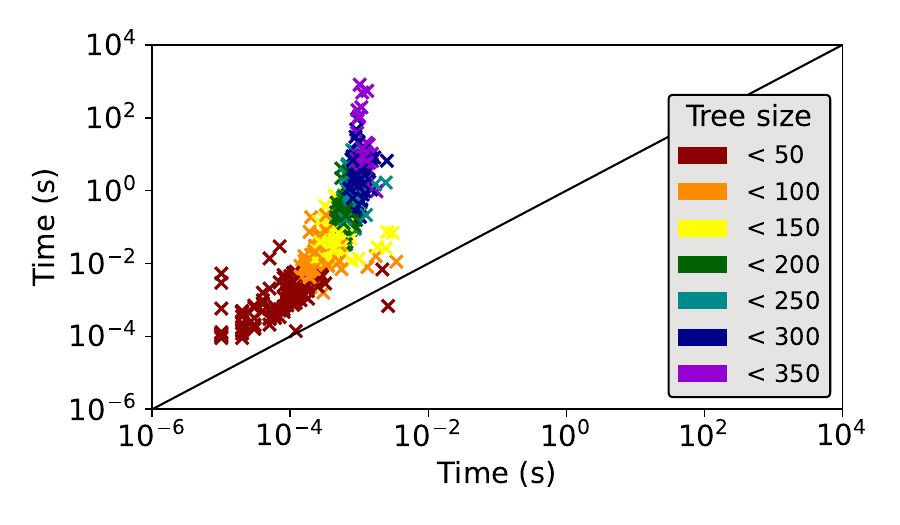}
        \caption*{\footnotesize (c) \texttt{BU}, $\texttt{BDD}_{\texttt{BU}}$}
    \end{minipage}
    
    \caption{ The first algorithm is the vertical axis while the second is the horizontal axis. The run time is in seconds, and each $\times$ represents a random ADT. For plots involving \texttt{BU}, only tree-structured ADTs are generated. Random ADTs were adjusted in size and number to ensure no run time exceeds $10^4$ seconds. All plots are based on 120 random ADTs with $|N| < 45$.}
    \label{fig:scatter}
    \vspace{-1.1em}
\end{figure}

\begin{figure}
    \centering
    \includegraphics[width=\columnwidth]{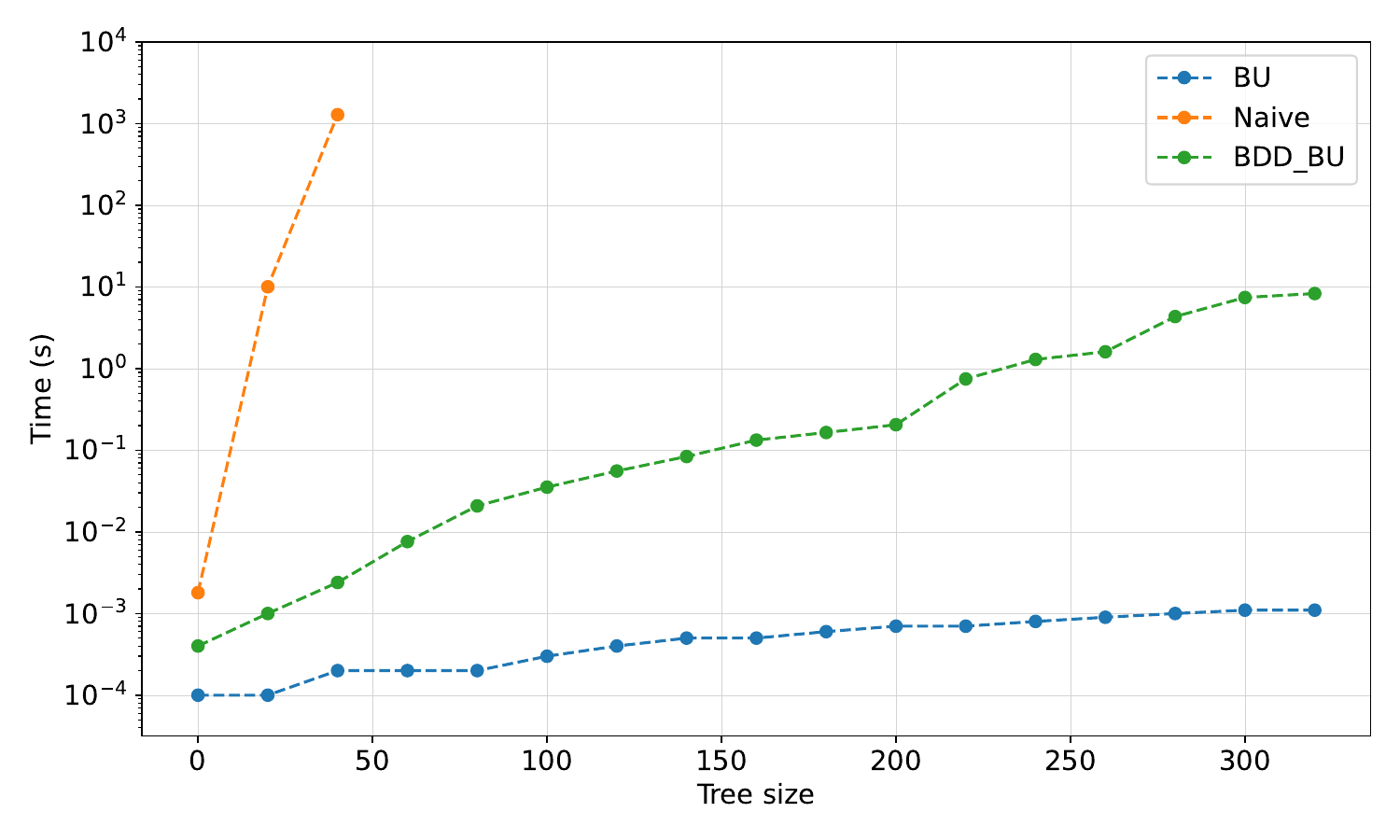}
    \caption{Summary of all the pairwise comparisons. The vertical axis represents the median time in seconds of each algorithm for ADTs grouped by the number of nodes $|N|$ at intervals of size 20. }
    \label{fig:pairwise-summary}
    \vspace{-1.1em}
\end{figure}

As both \texttt{BU} and $\texttt{BDD}_{\texttt{BU}}$ are quite fast, we extended our analysis for larger trees with up to 325 nodes. Remarkably, while the performance gap between the two remains consistent for trees with fewer than 100 nodes, this drastically changes as trees gain more nodes. For trees ranging from 300 to 325 nodes $\texttt{BDD}_{\texttt{BU}}$ requires approximately $10^3$ seconds, whereas $\texttt{BU}$ roughly $10^{-3}$ seconds.

In summary, this approach indicates that \texttt{BU} computes the Pareto Front the fastest for tree-structured ADTs, while $\texttt{BDD}_{\texttt{BU}}$ is the most efficient for DAGs. Furthermore, computation time remains low (below 10s) even for DAG-shaped ADTs an order of magnitude larger than our case study of Section \ref{ssec:case}.

%% file: Sections/6-conclusion.tex
\section{Conclusion} \label{Conclusion}

In this paper, we proposed a novel framework that incorporates defense metrics into ADTs, and presented efficient algorithms for computing the Pareto front between defense and attack metrics. The highlights behind this framework include a formal syntax and semantic model for representing ADT with attacker and defender attribute domains. We delved into \texttt{Bottom-Up} analysis and Binary Decision Diagram methods for computing the Pareto Front of the defender's and attacker's cost. We evaluated the performance of these approaches on a test suite consisting of randomly generated ADTs of sizes up to 325 nodes and observed significant variations in the speed of the algorithms.
In scenarios where the ADT has a tree structure, \texttt{BU} performs by far the best, with an average processing time of less than 1 second, even for trees with several hundred nodes. On the other hand, in cases where the ADT has a DAG structure and the attribute domains are absorbing semirings, then $\texttt{BDD}_{\texttt{BU}}$ is the next fastest choice, capable of analyzing trees up to 325 nodes under 30 minutes. 

As future works, one possible extension is incorporating probabilistic failures, as explored by Aslanyan and Nielson \cite{Aslanyan}, to enable scenarios where attackers and defenders operate under uncertainty. Attack-Fault-Defense Trees (AFDTs), as introduced in \cite{afdt_paper}, could provide a suitable framework for such extensions.
Another possible extension is enhancing algorithm efficiency through modular decomposition techniques. By dividing large ADTs into smaller modules, the computational cost of analyzing complex systems could be significantly reduced. Furthermore, optimizing BDDs by identifying orderings that minimize their size while retaining the defense-first property remains an important challenge. Developing automated methods for finding such orderings would improve scalability and enable efficient analysis of larger ADTs.
Finally, an interesting avenue for future research is extending our approach to ADTs with dynamic behaviour, similar to dynamic attack trees \cite{Jhawar2015}. The time-dependency between BAS/BDS could be modelled by their relative position in the BDD variable order.